\documentclass[11pt]{article}

\usepackage{fullpage}
\usepackage{appendix}
\usepackage{amsmath, amssymb, amsfonts}
\usepackage{calc}
\usepackage{bbm}
\usepackage{graphicx}
\usepackage{caption}
\usepackage{subcaption}
\graphicspath{{./figures/}}
\usepackage{authblk}

\usepackage[normalem]{ulem}
\usepackage[usenames,svgnames]{xcolor}

\usepackage{hyperref}
\hypersetup{
     colorlinks=true,           
     linkcolor=Salmon,          
     citecolor=blue,            
     filecolor=blue,            
     urlcolor=cyan,             
}

\newcommand{\Eq}[1]{Eq.~(\ref{#1})}
\newcommand{\Ineq}[1]{Ineq.~(\ref{#1})}
\newcommand{\Def}[1]{Definition~\ref{#1}}
\newcommand{\Lem}[1]{Lemma~\ref{#1}}
\newcommand{\Thm}[1]{Theorem~\ref{#1}}
\newcommand{\Sec}[1]{Sec.~\ref{#1}}
\newcommand{\cc}[1]{~\cite{#1}}
\newcommand{\cRef}[1]{Ref.~\cite{#1}}
\newcommand{\cRefs}[1]{Refs.~\cite{#1}}
\newcommand{\Fig}[1]{Fig.~\ref{#1}}

\DeclareMathOperator{\Span}{span}

\DeclareMathOperator{\Spec}{spec}
\DeclareMathOperator{\Eig}{Eig}
\DeclareMathOperator{\supp}{Supp}

\newcommand{\Tr}[1]{\mathrm{Tr}\left[#1\right]}
\DeclareMathOperator*{\tr}{Tr}
\newcommand{\Ptr}[2]{\mathrm{Tr}_{#2}\left[#1\right]}

\newcommand{\image}{\textrm{Im}}
\newcommand{\poly}{\mathrm{poly}}

\newcommand{\EqDef}{\stackrel{\mathrm{def}}{=}}
\newcommand\tab[1][1cm]{\hspace*{#1}}

\newcommand{\suppress}[1]{}
\newcommand{\myvec}[1]{\ensuremath{\mathsf{vec}}\left({#1}\right)}

\newcommand{\Id}{\mathbbm{1}}

\newcommand{\mcH}{\mathcal{H}}

\newcommand{\mcL}{\mathcal{L}}

\newcommand{\mcD}{\mathcal{D}}

\newcommand{\gs}{\Omega}
\newcommand{\im}{\mathrm{Im}}

\newcommand {\br} [1] {\ensuremath{ \left( #1 \right) }}
\newcommand {\minusspace} {\: \! \!}
\newcommand {\fn} [2] {\ensuremath{ #1 \minusspace \br{ #2 }}}
\newcommand{\norm}[1]{{\| #1 \|}}  
  
\newcommand{\ket}[1]{{ |{#1} \rangle }}  
\newcommand{\bra}[1]{{\langle {#1} | }}
\newcommand{\braket}[1]{{ \langle {#1} \rangle }}  
\newcommand{\ketbra}[2]{{ |{#1} \rangle\langle {#2} | }}
\newcommand{\dket}[1]{|#1\rangle\!\rangle}

\newcommand{\dmax}[2]{\fn{\mathrm{D}_{\max}}{#1 \middle\| #2}}
\newcommand{\dmin}[2]{\fn{\mathrm{D}_{\min}}{#1 \middle\| #2}}
\newcommand{\drel}[2]{\fn{\mathrm{D}}{#1 \middle \| #2}}
\newcommand{\mutinf}[2]{\fn{\mathrm{I}}{#1  :  #2}}
\newcommand{\maxinf}[2]{\fn{\mathrm{I}_{\max}}{#1 : #2}}
\newcommand{\maxinfeps}[2]{\fn{\mathrm{I}^\epsilon_{\max}}{#1 : #2}}
\newcommand{\maxinfdel}[2]{\fn{\mathrm{I}^\delta_{\max}}{#1 : #2}}
\newcommand{\SR}[2]{\fn{\mathrm{SR}}{#1 : #2}}
\newcommand{\ent}[1]{\fn{\mathrm{S}}{#1}}

\newcommand{\srho}[1]{\rho^{(#1)}}
\newcommand{\st}[1]{t^{(#1)}}
\newcommand{\ssigma}[1]{\sigma^{(#1)}}
\newcommand{\stau}[1]{\tau^{(#1)}}

\newcommand{\bigO}[1]{\ensuremath{\operatorname{O}\left(#1\right)}}
\newcommand{\bigtO}[1]{\ensuremath{\tilde{\operatorname{O}}\left(#1\right)}}
\newcommand{\bOmega}[1]{\ensuremath{\operatorname{\Omega}\bigl(#1\bigr)}}

\newcommand{\bTheta}[1]{\ensuremath{\operatorname{\Theta}\bigl(#1\bigr)}}

\newtheorem{theorem}{Theorem}[section]  
\newtheorem{remark}{Remark}[theorem]  
  
\newtheorem{definition}[theorem]{Definition}  
  
\newtheorem{claim}[theorem]{Claim}  
\newtheorem{lemma}[theorem]{Lemma}

\newtheorem{corol}[theorem]{Corollary}  
\newtheorem{fact}[theorem]{Fact}

\newcommand{\qedsymb}{\hfill{\rule{2mm}{2mm}}}  
\newenvironment{proof}[1][]{\begin{trivlist}  
  \item[\hspace{\labelsep}{\bf\noindent Proof#1:\/}]}
  {\qedsymb\end{trivlist}}

\newcommand{\eps}{\epsilon}  
\DeclareMathOperator{\mySR}{SR}

\title{Area laws and tensor networks for maximally mixed ground
states}

\author[1]{Itai Arad}
\author[2]{Raz Firanko}
\author[1,3,4]{Rahul Jain}

\affil[1]{Centre for Quantum Technologies, Singapore}
\affil[2]{Physics Department, Technion, Haifa 3200003, Israel
\footnote{Email: \texttt{razf680@campus.technion.ac.il}}}
\affil[3]{Department of Computer Science, 
  National University of Singapore}
\affil[4]{MajuLab, UMI 3654, Singapore 
  \footnote{Email:  \texttt{rahul@comp.nus.edu.sg}}}

\begin{document}
\maketitle

\begin{abstract}
  We show an area law in the mutual information for the
  maximally-mixed state $\gs$ in the ground space of general
  Hamiltonians, which is independent of the underlying ground space
  degeneracy. Our result assumes the existence of a `good'
  approximation to the ground state projector (a good AGSP), a
  crucial ingredient in previous area-law proofs.  Such approximations
  have been explicitly derived for 1D gapped local Hamiltonians and
  2D frustration-free locally-gapped Hamiltonians. As a
  corollary, we show that in 1D gapped local Hamiltonians, for any
  $\eps>0$ and any bi-partition $L\cup L^c$ of the system,
  \begin{align*}
    \maxinfeps{L}{L^c}_{\gs} \le \bigO{ \log (|L|\log(d))+\log(1/\eps)},
  \end{align*}
  where $|L|$ represents the number of sites in $L$, $d$ is the dimension of a site and
  $\maxinfeps{L}{L^c}_{\gs}$ represents the $\eps$-\emph{smoothed
  maximum mutual information} with respect to the $L:L^c$ partition
  in $\gs$. From this bound we then conclude $\mutinf{L}{L^c}_\gs
  \le \bigO{\log(|L|\log(d))}$ -- an area law for the mutual information in
  1D systems with a logarithmic correction. In addition, we show
  that $\gs$ can be approximated in trace norm up to $\eps$ with
  a state of Schmidt rank of at most $\poly(|L|/\eps)$, leading to
  a good MPO approximation for $\gs$ with polynomial bond
  dimension. Similar corollaries are derived for the mutual
  information of 2D frustration-free and locally-gapped local
  Hamiltonians.
\end{abstract}

\section{Introduction}
\label{sec:intro}

Understanding the structure of entanglement and correlations in
many-body quantum systems is a central problem in the theory of
condensed matter physics and quantum field theory. Properties of
this structure characterize different phases of matter and
transitions between them. From a computational point of view, the
amount of entanglement and correlations in many-body quantum
systems influences their computational complexity. For example, low
entanglement in a many-body quantum state can often be used to
construct an efficient classical representation of it. 

A useful method to characterize the amount of entanglement in a
many-body quantum state is by looking at the scaling behavior of the
\emph{mutual information} between a sub-region and the rest of the
system. This reduces to the \emph{entanglement entropy} of the
region when the underlying quantum state is pure. For random quantum
states, this quantity scales like the \emph{volume} of the region,
which saturates its maximal value. However, in many physically
interesting states, such as the ground states of local Hamiltonians,
mutual information and entanglement entropy often obey the so-called
\emph{area-law} behavior\cc{ref:Eisert2008AL}. In such cases, these
quantities scale like the surface area of the boundary between the
region and the rest of the system --- corresponding to a much lower
amount of correlations and entanglement.

Area laws are known to hold in several physically important states.
In particular, they have been shown to exist in Gibbs states
$\rho_G(\beta) \EqDef e^{-\beta H}/Z_\beta$, where $H$ is a local
Hamiltonian on a finite-dimensional lattice, $\beta$ is a finite
inverse temperature $\beta=\frac{1}{T}$ and $Z_\beta\EqDef \tr
e^{-\beta H}$ is a normalization factor, also known as the partition
function. In \cRef{ref:Cirac2008-Gibbs-AL} it has been shown that
for any region $L$ in the lattice and its complement region $L^c$,
the mutual information between $L$,$L^c$ is bounded by
$\mutinf{L}{L^c} = O(\beta\cdot|\partial L|)$. Therefore, when
$\beta=O(1)$, such Gibbs states satisfy an area-law in their mutual
information.

When the temperature goes to zero (equivalently, $\beta\to \infty$),
$e^{-\beta H}$ becomes proportional to the ground space projector
$\Pi_{gs}$, and the Gibbs state becomes the maximally-mixed state in
the ground space $V_{gs}$, which we call the \emph{maximally-mixed
ground-state}:
\begin{align}
  \gs \EqDef \Pi_{gs}/\tr\Pi_{gs}
  = \lim_{\beta\to \infty} e^{-\beta H}/Z_\beta .
\end{align} 
While the bound $\mutinf{L}{L^c} = O(\beta\cdot|\partial L|)$ from
\cRef{ref:Cirac2008-Gibbs-AL} becomes trivial in this limit, it is
often true that $\mutinf{L}{L^c}$ remains small. More precisely, when
the underlying Hamiltonian has a finite spectral gap and a
non-degenerate ground state, it is conjectured that its ground state
satisfies an area law of entanglement entropy. This is known as the
\emph{area-law conjecture}\cc{ref:Eisert2008AL}. This conjecture was
first shown to hold in non-interacting, relativistic field
theories\cc{ref:Bombelli1986-AL, Srednicki1993-AL}, as well as in
several exactly solvable models\cc{ref:Audenaert2002-AL,
ref:Vidal2003-AL}.  It was then rigorously proven for 1D
systems\cc{ref:Hastings2007, ref:ALV2012-AL, ref:Arad2013-1DAL}
using very different methods than the one used in the Gibbs state
case\cc{ref:Cirac2008-Gibbs-AL}. Finally, it was also proven for
higher dimensional lattices under various additional
assumptions\cc{ref:Masanes2009-AL,
ref:Cho2014-AL,ref:Brandao2015-AL, ref:anshu2022-2DAL}.

Over the past decade, the area-law conjecture was the subject of an
intensive research aimed at expanding the set of systems for which
it is shown to hold. A central challenge is of course to fully prove
it in 2D or higher dimensions without additional assumptions.
Another important line of research is to understand its validity in
the presence of \emph{ground space degeneracy}. To what extent do
all states in the ground space satisfy an area-law? How does the
bound depend on the ground space dimension? 

Already from the first proofs of the 1D area
law\cc{ref:Hastings2007, ref:ALV2012-AL, ref:Arad2013-1DAL}, it was
evident that as long as the ground state degeneracy is constant, one
can find a basis of ground states that satisfy an area-law (see also
\cRefs{ref:Huang2014-alg, ref:Chubb2015-computing}). This result was
further strengthened in \cRef{ref:arad2017rigorous} and then in
\cRef{ref:abrahamsen2020deg-AGSP} for ground spaces with higher
degeneracy. There it was shown that if the ground space degeneracy
is $r=\dim V_{gs}$, then for every $\ket{\psi}\in V_{gs}$ the
bi-partite entanglement entropy across any cut is upper bounded by
$O(\log r)$ (where we have taken the spectral gap and the local
Hilbert dimension to be $O(1)$). It is easy to verify that the $r$
scaling of this bound is optimal: for example, one can construct a
1D classical local Hamiltonian with $r=2^{O(n)}$ and find within
this subspace states with entanglement entropy of $O(n)$ across a
cut in the middle of the system.

The above discussion implies that in the high-degeneracy regime, not
all ground states necessarily obey an area law. However, of all the
ground states, there is one state of central importance, which is
$\gs$ --- the maximally mixed ground state.  In this paper, we
extend the AGSP (Approximate Ground Space Projector)
framework\cc{ref:ALV2012-AL, ref:Arad2013-1DAL,
ref:arad2017rigorous, ref:abrahamsen2020deg-AGSP,
ref:anshu2022-2DAL}, which is used to prove area-laws for pure ground states, to the case
of maximally mixed ground states. Specifically, we show that if
there exists a good AGSP for the Hamiltonian, in the sense of a 
favorable scaling between its closeness to the exact ground state
projector and its Schmidt rank (this shall be defined precisely in
\Def{def:AGSP} below), then, \emph{the maximally mixed ground state
satisfies an area-law in the mutual information regardless of the
ground space degeneracy}.  Good AGSPs are known to exist for
gapped 1D systems\cc{ref:Arad2013-1DAL}, as well as for 2D systems
that are frustration-free and locally
gapped\cc{ref:anshu2022-2DAL}.

Our results are in fact stronger; we show that when a good AGSP exists, then for
every contiguous set of qudits $L$ on the lattice, the
\emph{$\eps$-smoothed maximum information $\maxinfeps{L}{L^c}$},
which is closely related to the mutual information (see precise
definition in \Sec{sec:overview}) also satisfies an area-law.
Finally, we use that result to show that in 1D, for every
$\eps>0$, there exists a state $\gs_\eps$ such that $\norm{\gs -
\gs_\eps}_1\le \eps$ and $\mySR(\gs_\eps) =
O\big(\poly(|L|/\eps)\big)$ where $\mySR(\cdot)$ is the Schmidt rank
of $\gs_\eps$ (see Definition~\ref{def:op-SR} for an exact
statement) and $|L|$ is the size of the set $L$. Using this, we
construct a tensor-network approximation $\Psi$ with bond dimension
$\poly(n/\eps)$ for the maximally-mixed ground state in 1D,
for which $\norm{\Psi-\gs}_1\le \eps$.

The structure of the highly degenerate maximally-mixed ground-state
is interesting in several aspects. 
First, as $\gs$ is the zero temperature Gibbs state, for which our result establishes an area law using a good AGSP, it is natural to ask whether such AGSP-based techniques could be extended to derive area laws for Gibbs states at arbitrary temperatures. While mutual information area laws for finite-temperature Gibbs states of local Hamiltonians are already known\cc{ref:Cirac2008-Gibbs-AL}, it is unclear whether the existence of a good AGSP is sufficient to imply an area law at \textbf{all} temperatures, particularly for general gapped Hamiltonians. Second, the maximally-mixed state
is important from an information-theoretic point of view. It is
proportional to the ground-state projector, and therefore if it
satisfies an area-law, this implies non-trivial locality properties
of that operator. In particular, in 1D, our results can be phrased
as saying that the existence of a good AGSP --- which is essentially
a low entanglement operator that is a good approximation to the
ground state projector in the $L_\infty$ norm --- implies a good
matrix-product-operator (MPO) approximation in $L_1$ norm to the
maximally-mixed state that corresponds to that
ground state projector.
Finally,
exponentially degenerate ground spaces appear naturally in
Hamiltonian quantum complexity. For example, frustration-free
Hamiltonians that satisfy the conditions of the quantum Lovász local
lemma\cc{ref:Ambians2012-QLLL} have an exponential ground-state
degeneracy. In addition, such ground spaces might be relevant also
for understanding the structure of ground states in 2D
frustration-free systems. For such systems, we may consider the
partial Hamiltonian on a row, or on a column. Its ground space will
generally have an exponential degeneracy, and the global ground
state will be the intersection of these spaces. Understanding the
locality of the projectors into ground spaces of these partial
Hamiltonian (which, as we noted are proportional to their
corresponding maximally mixed state) can be useful for understanding
the global ground states.

Our proof enhances the AGSP framework\cc{ref:ALV2012-AL,ref:Arad2013-1DAL} for proving ground
state area laws with powerful tools from quantum information. In
\cRefs{ref:ALV2012-AL, ref:Arad2013-1DAL} an AGSP was used inside a
simple bootstrapping argument: it was shown that if there exists a
good AGSP, then there exists a product state with a large
overlap with the ground state --- a small
$\dmin{\gs}{\sigma_L\otimes \sigma_{L^c}}$ in the quantum
information terminology. Here we use the AGSP in a more elaborate
bootstrapping argument to upper-bound $\maxinf{L}{L^c}_{\gs'}$ of a
state $\gs'$ that is $\eps$-close to $\gs$. This provides us with a
bound on the maximal smooth information $\maxinfeps{L}{L^c}_\gs$ of
the maximally mixed ground state, from which the bound on
$\mutinf{L}{L^c}_\gs$ can be deduced using the continuity of the
mutual information and the fact that the max information upperbounds
it.

The structure of this paper is as follows. In \Sec{sec:results} we
give an exact statement of our results, together with the definition
of necessary measures of quantum information on which they rely. In
\Sec{sec:overview} we give an overview of our proof. In
\Sec{sec:background} we provide the necessary mathematical
background and preliminary results for the proofs. Finally, in
\Sec{sec:proofs} we give the full proof.

\section{Statement of the results}
\label{sec:results}

We consider a geometrically local Hamiltonian system $H=\sum_i h_i$,
defined on a finite-dimensional lattice. We assume the system is
made of $n$ qudits (spins) of local dimension $d$.  We let $V_{gs}$
denote the ground space of $H$, and $\Pi_{gs}$ be the projector into
$V_{gs}$. Finally, we denote the maximally-mixed state in $V_{gs}$
(i.e., the maximally-mixed ground state) by $\gs \EqDef \Pi_{gs}/r$,
where $r\EqDef \tr(\Pi_{gs}) = \dim(V_{gs})$ is the degeneracy of
the ground space.

To state our main result, we need the notion of an 
\emph{Approximate Ground Space Projector} (AGSP), which is a key
ingredient in many recent area-law proofs\cc{ref:ALV2012-AL,
ref:Arad2013-1DAL, ref:arad2017rigorous, ref:anshu2022-2DAL,
ref:abrahamsen2020deg-AGSP, abrahamsen2020sub}. Intuitively, by
working with an AGSP, we trade the accuracy of our approximation
for a good control of its locality. This translates to a tradeoff
between how close we approach the ground space and how much
entanglement we create on the way there. These two quantities are
characterized by the parameters $D$ and $\Delta$ that constitute a
$(D,\Delta)$-AGSP\footnote{Throughout this work, ‘’$\preceq$’’ denotes the standard operator ordering (see \Sec{sec:background}).}:
\begin{definition}[A $(D,\Delta)$-AGSP]
\label{def:AGSP}
  Let $H$ be a local Hamiltonian defined on some finite dimensional
  lattice with a ground state projector $\Pi_{gs}$, and let $L\cup
  L^c$ be a bi-partition of this lattice. For an integer $D\ge 1$
  and parameter $\Delta \in [0,1]$, an operator $K$ is called a
  $(D,\Delta)$-approximate ground state projector (AGSP) for
  $\Pi_{gs}$ with respect to the bi-partitioning $L\cup L^c$ if
  \begin{enumerate}
    \item \label{AGSP:bul1} $K\Pi_{gs} = K^\dagger\Pi_{gs}= \Pi_{gs}$.
  
    \item \label{AGSP:bul2} $K(\Id-\Pi_{gs}) K^\dagger  
      \preceq \Delta (\Id -\Pi_{gs})$.
    
    \item \label{AGSP:bul3} $K$ can be written as $K=\sum_{i=1}^D X_i 
      \otimes Y_i$, where $X_i \in \mcL(\mcH_L)$ and $Y_i
      \in \mcL(\mcH_{L^c})$. 
  \end{enumerate}
\end{definition}
Intuitively, $D$, which is the Schmidt rank of $K$, characterizes
how much entanglement it creates, and $\Delta$ tells us how quickly
it takes us to the ground space. Note that in \cRefs{ref:ALV2012-AL,
ref:Arad2013-1DAL} the condition on $\Delta$ was formulated as
$\norm{K\ket{\gs^\perp}}^2\le \Delta$ for every normalized vector
$\ket{\gs^\perp}$ that is perpendicular to the ground space. It is
easy to see that this, combined with \ref{AGSP:bul1}, is equivalent
to the condition $K(\Id-\Pi_{gs}) K^\dagger \preceq \Delta (\Id
-\Pi_{gs})$ above.

In \cRefs{ref:ALV2012-AL, ref:Arad2013-1DAL} it was shown that in
the case of a unique ground state, the existence of a
$(D,\Delta)$-AGSP with $D\cdot \Delta<1/2$ implies an upper-bound
$O(\log D)$ on the entanglement entropy. This was called the
\emph{bootsrapping lemma}. The problem of proving an area-law was
therefore reduced to the task of finding a ``good AGSP'' in which
$D\cdot \Delta <1/2$ and $\log(D) = O(|\partial L|)$, where
$\partial L$ is the boundary between $L, L^c$. Such AGSPs were found
for general local 1D systems with a global gap\cc{ref:ALV2012-AL,
ref:Arad2013-1DAL} and more recently also for 2D frustration-free
systems that are locally gapped\cc{ref:anshu2022-2DAL}. To a large
extent, our main result is a bootstrapping lemma for the maximally
mixed ground state, which shows how a good AGSP implies an area-law
for that state.

To state our results, we will also need to define some
generalizations of the notion of quantum relative entropies and
mutual information, which are commonly referred to as ``min-max
relative entropies''\cc{ref:Datta2009-max-ent}. To this aim, we
shall denote the set of quantum states over a Hilbert space $\mcH$
by $\mcD(\mcH)$, which is the convex set of Hermitian operators
$\rho$ over $\mcH$ with $\tr(\rho)=1$ and $\rho\succcurlyeq 0$. We
will also let $\mcD_-(\mcH)$ denote the set of \emph{sub-states},
for which the $\tr(\rho)=1$ requirement is relaxed to $\tr(\rho)\in
(0,1]$. For any (sub-)state $\rho$ we let $\image(\rho)$ denote its
image subspace and $\Pi_\rho$ the projector into that subspace. 
The min-max relative entropies are defined as follows
\begin{definition}[min-max relative entropies]
  Let $\rho,\sigma\in \mcD(\mcH)$ such that $\image(\rho)\subseteq
  \image(\sigma)$. We define,
  \begin{enumerate}
  \item \label{def:ent} Entropy of $\rho$: 
      \begin{align*}
      \ent{\rho} \EqDef - \Tr{\rho\log(\rho)}.
  \end{align*}
    \item \label{def:rel-ent} Relative entropy of $\rho$ with respect 
      to $\sigma$:
      \begin{align*} 
        \drel{\rho}{\sigma} \EqDef \Tr{\rho\log(\rho)} 
          - \Tr{\rho\log(\sigma)}.
      \end{align*}
                
      \item \label{def:max-rel-ent} Max relative entropy of $\rho$ 
        with respect to $\sigma$:
        \begin{align*}
          \dmax{\rho}{\sigma} 
            \EqDef \min \{\log t\in \mathbb R \; ; \; \rho 
              \preceq t \sigma \}.
        \end{align*}

      \item \label{def:min-rel-ent} Min relative entropy of $\rho$ 
        with respect to $\sigma$:
        \begin{align*} 
          \dmin{\rho}{\sigma} \EqDef 
            -\log\left(\Tr{\Pi_\rho \cdot \sigma }\right).
         \end{align*} 
\end{enumerate}
\end{definition}
We note that definitions \ref{def:max-rel-ent} and
\ref{def:min-rel-ent} above can be naturally generalized
to the case where $\rho \in \mcD_-(\mcH)$. For more information, see
\cRefs{ref:Datta2009-max-ent,ref:Berta2011quantum}. 

With these definitions at hand, we
define the corresponding mutual information measures as
follows:
\begin{definition}
  Let $\mcH = \mcH_L \otimes \mcH_R $ and $\rho_{LR} \in \mcD(\mcH)$. 
  Define,
  \begin{enumerate}
    \item \label{def:MI} Mutual information:
      \begin{align*} 
        \mutinf{L}{R}_{\rho} \EqDef \min_{\substack{\sigma_L \in
    \mcD(\mcH_L) \\ \sigma_R \in \mcD(\mcH_R)}}
            \drel{\rho}{\sigma_L\otimes \sigma_R}.
      \end{align*}
      
     \item \label{def:max-MI} Max mutual information:
       \begin{align*} 
         \maxinf{L}{R}_{\rho}\EqDef \min_{\substack{\sigma_L \in
    \mcD(\mcH_L) \\ \sigma_R \in \mcD(\mcH_R)}}
             \dmax{\rho}{\sigma_L\otimes \sigma_R}.
        \end{align*}

      \item \label{def:eps-max-MI} $\eps$-smoothed max mutual 
        information:
        \begin{align*} 
          \maxinfeps{L}{R}_\rho
            \EqDef \min_{\eta\in B_\eps(\rho)}
              \maxinf{L}{R}_\eta ,
        \end{align*}
        where $B_\eps(\rho)$ is the trace-norm ball around $\rho$,
    defined by:
        \begin{align*}
            B_\eps(\rho) \EqDef\{\eta\in \mcD_-(\mcH)\; ; \; 
                \norm{\rho-\eta}_1 \le \eps \} .
        \end{align*}
    \end{enumerate}
\end{definition}
Note that in the definition of $\maxinfeps{L}{R}_\rho$, the
minimization is over \emph{sub-states} that are $\eps$-close to
$\rho$. 
It is also worth noting that there are several ways to define the max mutual information, depending on whether the minimization is performed over $L$ or $R$, while fixing the other register as the marginal. However, these definitions are in fact equivalent (see \cRef{ref:Renner2013smooth}). The definition presented above is the suitable choice for the purpose of this work.
Also note that in the definition of mutual information the
minimum is obtained by taking $\sigma_L\otimes\sigma_R =
\rho_L\otimes\rho_R$, which yields the familiar formula
$\mutinf{L}{R}_{\rho}=\drel{\rho}{\rho_L\otimes \rho_R} =
\ent{\rho_L} + \ent{\rho_R} - \ent{\rho_{LR}}$. The same relation,
however, does not hold for $\maxinf{L}{R}_\rho$. Finally, also here
definitions~\ref{def:max-MI} and \ref{def:eps-max-MI} allow for
$\rho\in\mcD_-(\mcH)$.

We are now ready to state our main result, which is a bootstrapping
result for the $\eps$-smoothed max mutual-information.
\begin{theorem}[Area law bootstrapping for the $\eps$-smoothed maximum
  information]\ \label{thm:bootstrapping} 
  
  Let $H=\sum_i h_i$ be a local Hamiltonian on some lattice with a
  bi-partition $L\cup L^c$, and $d_L$ denote the Hilbert space dimension of
  subsystem $L$.
  Let $\gs$ denote its maximally-mixed
  ground-state. 
  Given an $\eps>0$, assume that there exists a
  $(D,\Delta)$-AGSP with respect to the $L\cup L^c$ bi-partitioning
  such that $D^2\cdot \Delta \le c_0\cdot \left(\frac{\eps}{\log
  d_L}\right)^8$, with $c_0=10^{-16}$. Then,
  \begin{align} \label{eq:max-AL}
    \maxinfeps{L}{L^c}_{\gs} \le 2\log D 
      + 12\log \Big(\frac{\log d_L}{\eps}\Big) + c_1 ,
  \end{align}
  where $c_1\approx 76$ is a universal constant.
\end{theorem}
Taking $\eps=(\log d_L)^{-1}$ and using the continuity of mutual
information (Fact~\ref{fact:cont-of-D}), we can turn the above
result into the following bound on the mutual information
\begin{corol}[Bootstrapping for the mutual information]
\label{corol:mutual-info}
  Under the same conditions in \Thm{thm:bootstrapping}, if there
  exists a $(D,\Delta)$-AGSP $K$ with $D^2\cdot \Delta \le c_0\cdot
  (\log d_L)^{-16}$ then the mutual information in the maximally
  mixed ground state between $L$ and $L^c$ is upperbounded by
  \begin{align}
    \mutinf{L}{L^c}_{\gs} \le 2\log D + 24\log \log d_L + O(1) .
  \end{align}
\end{corol}

Using our bootstrapping results together with the 1D and 2D AGSP
constructions of \cRef{ref:Arad2013-1DAL} and
\cRef{ref:anshu2022-2DAL} (see \Sec{sec:AGSP}), we get the following area laws
\begin{corol}[Area law for the maximally mixed ground state in 1D]
\label{corol:1D-AL} 
  Let $H=\sum_{i=1}^{n-1} h_i$ be a 1D local Hamiltonian over
  qudits with a spectral gap $\gamma$, and let
  $\gs$ be its maximally-mixed ground-state.  For every contiguous
  segment $L$ in the 1D lattice and for every $\eps>0$ 
  such that\footnote{As expected from gapped local Hamiltonians on qudits of constant dimension.}
  $\gamma^{-1}\cdot\log^3(d/\gamma) = \bigO{\log(|L|\log (d)/\eps)} $, 
  \begin{align}\label{eq:1D-AL}
  \log(D) & = \bigO{\gamma^{-1/3}\cdot\log(|L|\log (d)/\eps)},\\
    \maxinfeps{L}{L^c}_{\gs} 
      & = \bigO{\gamma^{-1/3}\cdot\log(|L|\log (d)/\eps)},\\
      \mutinf{L}{L^c}_{\gs} & = \bigO{\gamma^{-1/3}
        \cdot\log(|L|\log (d))} .
  \end{align}
  where $|L|$ denotes the number of qudits in $L$
  and $D$ is the Schmidt rank of an AGSP that suites the conditions in \Thm{thm:bootstrapping}. 
  If we have $\gamma^{-1}\cdot\log^3(d/\gamma) = \bOmega{\log(|L|\log (d)/\eps)}$,
  then we would get that $\log(D)$, $\maxinfeps{L}{L^c}_{\gs}$ and $\mutinf{L}{L^c}_{\gs}$ are $\bigO{\gamma^{-1}\cdot \log^3(d/\gamma)}$.
\end{corol}

\begin{corol}[Area law for the maximally mixed ground state in 2D]
\label{corol:2D-AL} 
  Let $H$ be a 2D frustration-free local Hamiltonian on a
  rectangular lattice of $d$-dimensional qudits with $d=\bigO 1$ and
  an $O(1)$ local spectral gap, and let $\gs$ be its maximally mixed
  ground state.  Then for every bi-partitioning of the system
  $L:L^c$ along a vertical arc $\partial L$ and for every $\eps>0$,
  \begin{align}
    \maxinfeps{L}{L^c}_{\gs} & 
     = \bigO{\log(1/\eps) 
       \cdot |\partial L|^{1+\bigO{\log^{-1/5}|\partial L|}}},
  \end{align} 
  where $|\partial L|$ denotes the length of the
  boundary line $\partial L$. In addition,
  \begin{align}
    \mutinf{L}{L^c}_{\gs} = \bigO{ 
      |\partial L|^{1+\bigO{\log^{-1/5}|\partial L|}}} .
  \end{align}
\end{corol}
We remark that if a better AGSP is discovered in the future 
for 2D Hamiltonians so that $D^2\Delta\leq 1/2$ and
$\log(D)=\bigO{|\partial L|}$ instead of current results (given in
\Eq{def:2D-D} in \Sec{sec:AGSP}), our bootstrapping theorem would
yield $\maxinfeps{L}{L^c}_{\gs} = \bigO{\log(|L|/\eps) \cdot
{|\partial L|}}$ and $\mutinf{L}{L^c}_{\gs} = \bigO{ {|\partial
L|}\cdot\log|L|}$.

{~}

In addition to area-law bounds on the mutual information, we can
also use the $\eps$-smoothed maximum information bound to show the
existence of a low Schmidt-rank approximation for the
maximally-mixed ground state. 

We show that under the same
settings as in \Thm{thm:bootstrapping}, one can obtain an
approximation to the maximally-mixed ground state with a low
operator Schmidt rank (see \Def{def:op-SR}).
In fact, the tools that we introduce
enable us to prove an even stronger result: the maximally-mixed
ground state can be purified on a larger system for which there is a low
Schmidt-rank approximation due to any cut.
\begin{theorem}[Low Schmidt-rank approximation]
\label{thm:LowSR}  
   Let $\eps>0$, and let $H=\sum_i h_i$ be a local Hamiltonian on
  some lattice of qudits with a maximally-mixed ground state $\gs$.
   Suppose that for any bi-partition of the lattice
  $L\cup L^c$, there exists a $(D,\Delta)$-AGSP such that $D^2\cdot
  \Delta \le c_0\cdot \left(\frac{\eps}{\log d_L}\right)^8$ where 
  $d_L=\dim(\mcH_L)$ and $c_0$ is the universal constant from
  \Thm{thm:bootstrapping}. Then there exists an auxiliary system
  $E$ and a purification $\gs_{A}\mapsto \ket \gs_{A\tilde A E}$ 
  such that
  $\gs_A=\Ptr{\ketbra \gs \gs} {\tilde A E}$, and for any bi-partition of the lattice
  $A=L\cup L^c$, there is a state $\ket {\psi^{(L)}}_{A\tilde A E}$ for which:
  1.~$\norm{\ket{\gs}-\ket{\psi^{(L)}}}^2\le\eps$.
  2. The Schmidt rank of $\ket{\psi^{(L)}}_{A\tilde A E}$ with respect to the $L\tilde L: L^c\tilde{L^c} E$  bi-partition satisfies 
  \begin{align*}
      \mySR (\ket {\psi^{(L)}})\le 
    	  49 D^2 \cdot \Big(\frac{\log d_L}{\eps}\Big)^2.
  \end{align*}
\end{theorem}

Our final result is restricted to scenarios where the underlying lattice is $1D$. We derive an MPO
(matrix-product-operator) approximation to the maximally mixed
ground state. To construct such a tensor network, one needs to
project onto the largest Schmidt-states with respect to any cut in
the 1D lattice, while controlling the truncation error resulting
from each of these. The analysis of these sequential projections is
best suited to $L_2$ norm rather than the $L_1$ norm we have used so
far. For this reason, we prefer to work with the purification of the
state, given in \Thm{thm:LowSR}, instead of the original density operator.  Choosing
$\eps'=\eps/n$ and truncating sequentially with each cut, we get a
1D tensor network structure, with bond dimension which is
$\poly(n/\eps)$ as guaranteed by \Thm{thm:LowSR}. As a result, we
obtain a MPS (matrix-product-state) tensor network approximation to
the purification of the ground state, which results in an MPO after
tracing out the auxiliary systems (see \Fig{fig:MPO}).

\begin{corol}[An MPO
  approximation for the maximally-mixed ground-state in 1D]
  \label{corol:TN} Let $H=\sum_{i=1}^{n-1} h_i$ be a 1D local
  Hamiltonian of qudits with $d=\bigO{1}$ and an $O(1)$ spectral gap, and let $\gs$ be its
  maximally-mixed ground-state and $\eps>0$. Then there is a
  matrix-product-operator (MPO) state $\Psi$ with
  $\poly(n/\eps)$-bond dimension such that $\norm{\gs-\Psi}_1\le
  \eps$.
\end{corol}

The proofs of our main bootstrapping theorem and its following
corollaries are given in \Sec{sec:proofs}. \Thm{thm:LowSR} and
corresponding tensor-network is derived in \Sec{Sec:approx}.
\begin{remark}
  Note that the actual local Hamiltonian is never used
  directly in our proofs, but only
  implicitly for deriving a good AGSP (see Facts~\ref{fact:1D-AGSP}
  and \ref{fact:2D-AGSP}). One can 
  therefore state our results in terms of a normalized projector
  $\gs=\Pi/\Tr{\Pi}$ where $\Pi$ admits a good
  $(D,\Delta)$-approximation (as in \Def{def:AGSP}).
\end{remark}

\section{Overview of the proof of \Thm{thm:bootstrapping}}
\label{sec:overview}

Let us present the idea of the proof in the following scenario. Let
$\gs$ be the maximally-mixed ground state of a local Hamiltonian system
on a lattice of qubits, and consider a bi-partition of the system
into two parts, $L$ and $R$\footnote{We are changing the notation
here from $L\cup L^c$ to $L\cup R$ --- but this is merely to reduce
the clutter in our notation.}. By definition of the maximum
information, there exists a product state $\sigma_L\otimes\sigma_R$
such that $\gs \preceq t \sigma_L\otimes \sigma_R$, where
$t=2^{\maxinf{L}{R}_{\gs}}$ is the \emph{minimal} factor that is
needed to upperbound $\gs$ by a product state. Our goal is to
upperbound $t$. 

We now assume that there exists a $(D,\Delta)$-AGSP $K$, for which
$D^2\cdot\Delta \le 1/2$ (see
\Def{def:AGSP} in \Sec{sec:results} for a
formal statement) and apply it on both sides of the inequality $\gs \preceq
t \sigma_L\otimes \sigma_R$, to obtain
\begin{align}
\label{eq:overview-I}
    \gs = K\gs K^\dagger  \preceq K\sigma_L\otimes \sigma_R K^\dagger.
\end{align}
We now perform a procedure in the spirit of the bootstrapping lemma
from \cRef{ref:Arad2013-1DAL} adapted to mixed states and max
information. Using \Lem{lem:SR-Dmax} and the fact that the Schmidt
rank of $K$ is $D$, we upperbound the maximum information of
$K\sigma_L\otimes\sigma_R K^\dagger$ using a
product state $\tau_L\otimes \tau_R$ such that
\begin{align}
\label{eq:overview-II}
  K\sigma_L\otimes \sigma_R K^\dagger 
    \preceq \tr(K\sigma_L\otimes\sigma_R K^\dagger)\cdot D^2
      \cdot \tau_L\otimes \tau_R.
\end{align}
Using the fact that $K$ approximates the ground state projector, we
can upperbound the trace by decomposing
$K\sigma_L\otimes\sigma_R K^\dagger$ to the ground state part and 
 orthogonal part: $\tr(K\sigma_L\otimes\sigma_R K^\dagger) =
\tr(\Pi_{gs}K\sigma_L\otimes\sigma_R K^\dagger) +
\Tr{(1-\Pi_{gs})K\sigma_L\otimes\sigma_R K^\dagger}$. Since the
orthogonal part is shrunk by the AGSP by a factor of $\Delta$ and
$\Pi_{gs}$ is fixed by the AGSP, we get from \eqref{eq:overview-I}
and \eqref{eq:overview-II} that
\begin{align*}
  \gs \preceq t\cdot\big(\Tr{\Pi_{gs}\sigma_L\otimes\sigma_R}
    + \Delta\big)\cdot D^2 \cdot \tau_L\otimes \tau_R.
\end{align*}
The final step is to note that $t$ is the minimal factor that is
needed to upperbound $\gs$ by a product state, and therefore
necessarily $t\le t\cdot\big(\Tr{\Pi_{gs}\sigma_L\otimes\sigma_R}
+\Delta\big) \cdot D^2$. Assuming that $K$ is a ``good AGSP'' with
$D^2\Delta\le 1/2$, we conclude that
$\tr(\Pi_{gs}\sigma_L\otimes\sigma_R) \ge 1/2D^2$. 

To finish the proof we make the following crucial
assumption: suppose that $\sigma_L\otimes\sigma_R$ is a \emph{flat
state}, i.e. it is proportional to a projector on its support. Then
Fact~\ref{fact:Dmax_flat} tells us that
$\tr(\Pi_{gs}\sigma_L\otimes\sigma_R)=1/t$, which implies that 
$1/t \ge 1/2D^2$ and therefore $t\le 2D^2$.

Note, however, that the assumption of $\sigma_L\otimes \sigma_R$
being flat is hard to justify. Instead, another technique is
required to relate the ground state overlap of
$\sigma_L\otimes\sigma_R$ to the maximum information. For this, we
use the so-called ``brothers extension''\cc{ref:Anurag-Rahul-2016}
which extends $\gs$ and $\sigma_L\otimes\sigma_R$ to a larger
Hilbert space where the brothers extension of
$\sigma_L\otimes\sigma_R$ becomes flat, as specified in
\Lem{lem:flat}. This creates another problem though: the brothers
extension requires projecting out from $\gs$ contributions from
the small spectrum of $\sigma_L\otimes\sigma_R$, and hence
results in a density matrix $\rho$ that is $\delta$-close to $\gs$, but not
$\gs$ itself. To handle this, we generate (using a similar yet more
complicated procedure) a sequence of states $\{\srho k\}_k$ which are
in the $\eps$ ball of $\gs$, together with a sequence of positive
numbers $\st k$ and product states $\ssigma{k}$ such that $\srho
k\preceq \st k \ssigma k$. 
Using various techniques like the brothers (flat) extension and the quality of the AGSP,
one can relate the change in $\st{k}\mapsto\st{k+1}$ with the maximum information of $\srho k$ (\Eq{eq:next-t}).
Due to a saturation argument of the maximum information 
(\Lem{lem:Imax_bound}), we conclude that sequence $\{\st k\}$ should accumulate,
resulting in an upper bound on the maximum information within the $\eps$-ball around $\gs$.

\section{Preliminaries and mathematical background} 
\label{sec:background}

This section provides the information-theoretic preliminaries,
notations, definitions, facts, and lemmas needed to prove our main
result.
\subsection{States} We denote the Hilbert space of a system
$A$ with $\mcH_A$ and the dimension of $\mcH_A$ as $d_A$. Let the
set of linear operators on $\mcH_A$ be $\mcL(\mcH_A)$; the set of
states (density operators) on $\mcH_A$ be $\mcD(\mcH_A)$, and the
set of sub-states be $\mcD_-(\mcH_A) \EqDef \{\rho\in
\mcL(\mcH_A)\;|\; \rho\succeq 0, \Tr{\rho}\in (0,1] \}$. Let
$\norm{M}$ denote the operator (spectral) norm of the operator $M$,
and $\norm{M}_1$ denote the trace norm, i.e. $\norm{M}_1 \EqDef
\Tr{\sqrt{M^\dagger M}}$. Let $\Spec(M)$ denote the set of its
distinct eigenvalues. Let $\image(M)$ represent the image of an
operator $M$, $d_M \EqDef \dim(\im(M))$, and $\Pi_M$ represent the
projector onto $\image(M)$. Let $\Id$ represent the identity
operator. For $M \in \mcL(\mcH_L \otimes \mcH_R)$, its Schmidt rank
across the $L:R$ cut is denoted $\SR{L}{R}_M$. For operators $M, N$,
we write $M \succeq N$ to represent that $M-N \succeq 0$, that is $M-N$ is
positive semi-definite. 
We now extend the definition of Schmidt rank to operators:
\begin{definition}[Operator Schmidt rank]
\label{def:op-SR} 
  Let $X$ be an operator on a bi-partitioned system with a
  Hilbert space $\mcH_{AB}=\mcH_A\otimes \mcH_B$. Then the Schmidt
  rank of $X$ with respect to the $(A,B)$ bi-partition is defined by the minimal number of product operators
  needed to express it:
  \begin{align}
    \mySR(X) \EqDef \min \Big\{R \; ; \; \exists \{A_i, B_i\}
  \ \text{s.t.,}\  X = \sum_{i=1}^R A_i\otimes B_i \Big\} .
  \end{align}
\end{definition}

The following are three basic facts about states and measurements that
we shall use later in our proof.

\begin{fact} 
\label{fact:norm} 
  Let $\ket{v}, \ket{w}$ be unit vectors. Then,
  \begin{align*}
    \norm{\big(\ketbra{v}{v}-\ketbra{w}{w}\big)}_1
      = 2\sqrt{1 - |\braket{v|w}|^2}.
  \end{align*}
\end{fact}

\begin{fact}[Theorem III.4.4 in \cRef{ref:Bhatia-book}]
\label{fact:bhatia}
  Let $\rho, \sigma$ be states. Then,
  \begin{align*}
    \norm{\Eig^\downarrow(\rho)-\Eig^\downarrow(\sigma)}_1
      \le \norm{\rho - \sigma}_1,
  \end{align*}
  where $\Eig^\downarrow(\cdot)$ is the vector of non-increasing
  eigenvalues. 
\end{fact}

\begin{fact}[Gentle measurement Lemma\cc{ref:Winter-1999-gentle, 
ref:Ogawa2007-gentle}]
\label{fact:gentle} 
  Let $\rho \in \mcD_{-}(\mcH)$ and $\Pi$ be an orthogonal
  projection onto a subspace of $\mcH$. Then, 
  \begin{align*}
      \norm{\rho - \Pi \rho \Pi}_1 
        \le 2\sqrt{\Tr{(\Id - \Pi) \rho}}.
  \end{align*}
\end{fact}

\subsection{The \texorpdfstring{$\ensuremath{\mathsf{vec}}$}{vec} map} \label{sec:Vec}

Consider the map $\ensuremath{\mathsf{vec}}: \mcL(\mcH_A)
\rightarrow \mcH_A \otimes \mcH_{\tilde A}$:
\begin{align*}
  \forall v, w: \myvec{\ket{v}\bra{w}} \EqDef \ket{v} 
    \otimes \overline{\ket{w}},
\end{align*}
where $\overline{\ket{w}}$ is the entry-wise conjugate of $\ket{w}$
in the standard basis. The map satisfies the following properties.
\begin{fact}[see \cRef{ref:watrous-book}] 
\label{fact:vec}\ 
  \begin{enumerate}
    \item $\myvec{X + Y} = \myvec{X} + \myvec{Y}$.
    
     \item $\Tr{X^\dagger Y} = \myvec{X}^\dagger \cdot \myvec{Y}$.
     
     \item $\Ptr{\myvec{X} \myvec{X}^\dagger}{\tilde A} = X X ^\dagger.$
     
     \item For every non-negative $X$, the vector $\ket{\psi} 
       = \myvec{\sqrt{X}}$ is a purification of $X$.
   \end{enumerate}
\end{fact}
Through the manuscript, we use the notation $\dket{A}$ and $\myvec{A}$ interchangeably. 

We end with the following remark that relates distance between states and their matrix square roots:
\begin{fact} [Lemma 3.34 in \cRef{ref:watrous-book}] \label{fact:sqrt}
    Let $P_1, P_2$ be positive semi-definite operators, and let $\sqrt{P_1},\sqrt{P_2}$ be their respective matrix square roots. Then
    \begin{align*}
        \norm{\sqrt{P_1}-\sqrt{P_2}}_2^2 \le \norm{P_1-P_2}_1.
    \end{align*}
\end{fact}

\subsection{Entropies and information} 

Most of the entropic functions we use in the proof are defined in 
\Sec{sec:results}, where we present our exact result.  Here we
present some well-known facts and lemmas about these functions. 

We begin with the continuity of the mutual information, whose proof
can be found, e.g., in \cRef{ref:watrous-book}.
\begin{fact}[Continuity of mutual information]
\label{fact:cont-of-D} Let $\rho,\sigma\in\mcD(\mcH_{LR})$. Then,
  \begin{align*}
    |\mutinf{L}{R}_\rho-\mutinf{L}{R}_\sigma|
      \le \frac{3}{2} \cdot \log(d_L)\cdot\norm{\rho-\sigma}_1 + 3.
  \end{align*}
\end{fact}

In addition, we need the following bounds on the relative and mutual
information from their maximal counterparts. 
\begin{fact} \label{clm:D_inequality}
  Let $\rho,\sigma\in \mcD(\mcH)$ such that $\image(\rho)\subseteq
  \image(\sigma)$. We have,
  \begin{align*} 
    \drel{\rho}{\sigma} \le \dmax{\rho}{\sigma} ,\\
    \mutinf{L}{R}_\rho \le \maxinf{L}{R}_\rho .
\end{align*}
\end{fact}
\begin{proof}
  Let $t=2^{\dmax{\rho}{\sigma}}$. Then $\rho \preceq t\sigma$ and by
  the monotonicity of the operator logarithm\cc{ref:Bhatia-book},
  \begin{align*}
    \log(\rho)\preceq \log(t\sigma)=\log(t)\Id+\log(\sigma).
  \end{align*}
   This implies,
   \begin{align*}
        \drel{\rho}{\sigma}  &= \Tr{\rho\log(\rho)}
          -\Tr{\rho\log(\sigma)} \\
        &\le \Tr{\rho(\log(\sigma)+\Id\log(t))}
          -\Tr{\rho\log(\sigma)} \\
        &=\log(t),
   \end{align*}
   proving the first inequality. The second inequality follows
   from the first inequality and the definitions of mutual
   information and max mutual information.
\end{proof}

Additionally, we will make use of the following lemma.
\begin{fact}[Lemma B10 in \cRef{ref:Berta2011quantum}]
    \label{lem:Imax_bound}
    Let $\rho_{LR}$ be a sub-state in $\mcD_-(\mcH_{LR})$, then
    \begin{align*}
        \maxinf{L}{R}_\rho \le 2 \log \min\{d_L,d_R\}.
    \end{align*}
    where $d_L$ and $d_R$ are the dimensions of $\mcH_L$ and
    $\mcH_R$, respectively.
\end{fact}

\begin{lemma} [Small $\mathrm{D_{\max}}$ implies short distance]
\label{lem:short-dist}
  Let $\rho\in\mcD_-(\mcH)$ and $\sigma\in\mcD(\mcH)$ such that
  \begin{align*}
      \rho \preceq (1+\delta)\sigma.
  \end{align*}
  Then, $\norm{\rho-\sigma}_1\le 2\delta + (1-\Tr{\rho})$.
\end{lemma}
\begin{proof}
  Decompose $\rho-\sigma$ into positive and negative parts,
  $\rho-\sigma=(\rho-\sigma)_+-(\rho-\sigma)_-$, and define
  $p_\pm\EqDef \Tr{(\rho-\sigma)_\pm}$ such that
  $\norm{\rho-\sigma}_1=p_+ + p_-$. Note that 
  \begin{align*}
      p_+-p_-=\Tr{\rho-\sigma}=\Tr{\rho}-1,
  \end{align*}
  i.e. $p_-=p_+ + (1-\Tr{\rho})$. Now, subtract $\sigma$ from the
  operator inequality to achieve $\rho-\sigma \preceq \delta
  \sigma$, and therefore
  \begin{align}
    (\rho-\sigma)_+ \preceq \delta \sigma + (\rho-\sigma)_- .
  \label{eq:rho-sigma}
  \end{align}
  Let $P_+$ be the projection into the support of $(\rho-\sigma)_+$.
  Then $P_+\cdot (\rho-\sigma)_+\cdot P_+ = (\rho-\sigma)_+$ and 
  $P_+\cdot (\rho-\sigma)_-\cdot P_+ = 0$. Therefore multiplying 
  \eqref{eq:rho-sigma} by $P_+$ from both sides, we get
  \begin{align*}
    (\rho-\sigma)_+ \preceq \delta P_+ \sigma P_+ \quad
	\Rightarrow\quad p_+ \le \delta \Tr(P_+ \sigma) \le \delta .
  \end{align*}
  Therefore, 
  \begin{align*}
    \norm{\rho-\sigma}_1 = p_+ + p_-  = 2p_+ + (1-\Tr{\rho})
	\le 2\delta +(1-\Tr{\rho}).
  \end{align*}
\end{proof}

\subsection{Flat states}

Central objects in our proofs are \emph{flat states}, defined as
follows.
\begin{definition}[Flat state] \label{def:flatstate}
  We call a state $\tau$ flat if it is uniform in its support, that 
  is all its non-zero eigenvalues are identical.  Alternatively,
  it can be written as
  \begin{align*}
    \tau = \frac{\Pi_\tau}{d_\tau}, \qquad d_\tau = \tr(\Pi_\tau) .
  \end{align*}
\end{definition}
Flat states possess the following relation between the maximum
relative entropy and minimum relative entropy.  
\begin{fact} \label{fact:Dmax_flat}
  Let $\rho$ and $\sigma$ be flat states such that
  $\im(\rho)\subseteq \im(\sigma)$. Then
  \begin{align} \label{eq:Dmax_flat}
    \log \frac{1}{ \Tr{\Pi_\rho\sigma}} = \dmin{\rho}{\sigma}
      = \dmax{\rho}{\sigma} 
      = \log\left(\frac{d_\sigma}{d_\rho} \right).
  \end{align}
\end{fact}
\begin{proof}
  Since $\rho,\sigma$ are flat states, they can be written as 
  \begin{align*}
    \rho = \frac{1}{d_\rho} \Pi_{\rho} \quad ; \quad 
    \sigma = \frac{1}{d_\sigma} \Pi_{\sigma}.
  \end{align*}
  Consider,
  \begin{align*}
    \Tr{\Pi_\rho\sigma} = \frac{1}{d_\sigma}
      \mathrm{Tr}\big(\underbrace{\Pi_\rho\Pi_\sigma}_{=\Pi_\rho}
        \big) = \frac{d_\rho}{d_\sigma} = 2^{-\dmax{\rho}{\sigma}}.
  \end{align*}
\end{proof}

In our work, we use a slightly stronger version of this fact.
\begin{lemma} \label{lem:Dmax_flat_2}
    Let $\sigma$ be a flat state and $\rho$ be a sub-state such that $\im(\rho)\subseteq \im(\sigma)$.
    Then
  \begin{align*} 
      \dmax{\rho}{\sigma} 
      = \log\left(d_\sigma \cdot \norm{\rho} \right).
  \end{align*}
\end{lemma}
\begin{proof}
    We write $\sigma=\frac{\Pi_\sigma}{d_\sigma}$ and use Lemma B4 in \cRef{ref:Berta2011quantum} to write $\dmax{\rho}{\sigma}=\log(\norm{\sigma^{-1/2}\rho\sigma^{-1/2}})=
    \log(d_\sigma\norm{\rho})$.
\end{proof}

\subsection{Local Hamiltonians and Approximate Ground State Projectors}
\label{sec:AGSP}

In this work we consider local Hamiltonians defined on finite
dimensional lattices. Formally, let $\Lambda$ denote the sites of a
finite dimensional lattice, and assume that at each sites there is a
$d$-dimensional qudit (a spin) so that the total Hilbert space of
the system is $\mcH = \big(\mathbb{C}^d\big)^{\otimes |\Lambda|}$. A
$k$-body local Hamiltonian on $\Lambda$ is an operator of the form
$H=\sum_x h_x$ on $\mcH$ where the summation is over all
geometrically local subsets $x\subset \Lambda$ with $|x|\le k$. The
operators $\{h_x\}$ are hermitian and act non-trivially only on the
qudits in $x$, i.e., $h_x = \hat{h}_x\otimes\Id_{rest}$, where
$\hat{h}_x$ acts on the Hilbert space of the qudits in $x$.
Throughout this work we shall assume that $\norm{h_x}\le J$ for all
$x$ for some fixed energy scale $J$. In such cases we can always
pass to a dimensionless setup and assume without loss of generality
that $0\preceq h_x\preceq \Id$ for all $x$. 

Given a local Hamiltonian $H = \sum_x h_x$, we denote its
eigenvalues by $E_0\le E_1\le E_2\le \ldots$, and their
corresponding eigenspaces projectors by $\Pi_0, \Pi_1, \Pi_2,
\ldots$ so that $H = \sum_{i\ge 0} E_i\cdot \Pi_i$. The eigenvalues
$E_i$ are called \emph{energy levels}. The lowest eigenvalue of $H$
is called the \emph{ground energy} of $H$ and its corresponding
eigenspace is called the \emph{ground space} of $H$.  Every state in
the ground space is called a \emph{ground state}. In what follows,
we will denote the ground space projector by $\Pi_{gs}$.

The ground states of a local Hamiltonian are of a great interest for
physicists and chemists, as they determine important low temperature
properties of the underlying system. An important factor is the
\emph{spectral gap} of the system, $\gamma \EqDef E_1-E_0$, which is
the difference between the first excited energy level and the ground
energy. The presence of a large spectral gap can be associated with
a decay of correlations in ground states of the
system\cc{ref:Hastings2006-exp-decay}, as well as with area-law
bounds on the entanglement entropy of the ground
states\cc{ref:Eisert2008AL, ref:Bombelli1986-AL, Srednicki1993-AL,
ref:Audenaert2002-AL, ref:Vidal2003-AL, ref:Hastings2007,
ref:ALV2012-AL, ref:Arad2013-1DAL}. Very often, when we consider
a local Hamiltonian system, we actually consider a \emph{family} of
such systems with an increasing size $|\Lambda_n|$. In such case it
is customary to use the notation $\gamma=\Omega(1)$ to describe a
situation in which the spectral gap is lower bounded by a constant as
the system increases.

\emph{Frustration-free} local Hamiltonians are an important
sub-class of local Hamiltonians $H=\sum_x h_x$ in which the ground
space of $H$ is also a ground space of every individual $h_x$ term,
i.e, $h_x\Pi_{gs} = E^{(x)}_0\Pi_{gs}$ for all $x$. Frustration-free
Hamiltonians are important to our settings because they naturally
give rise to highly degenerate ground spaces. Moreover, the AGSP
results we import below are much simpler to present in the
frustration-free case in 1D\cc{ref:Arad2013-1DAL}, while in 2D they
are still lacking for frustrated Hamiltonians. As a final remark,
our results may even be generalized to cases where there are many
low-energy states with exponentially-close energy levels, which are
separated by a gap from the rest of the spectrum, as was studied in
\cRef{ref:arad2017rigorous}. We leave this intriguing possibility
for future work.

A powerful framework to prove an area-law in ground states of local
Hamiltonian is the so-called \emph{approximate ground state
projector} (AGSP) framework\cc{ref:ALV2012-AL, ref:Arad2013-1DAL,
ref:arad2017rigorous, ref:anshu2022-2DAL,
ref:abrahamsen2020deg-AGSP, abrahamsen2020sub}, which is formally
defined in \Def{def:AGSP}. As its name suggests, an AGSP is an
operator that approximates the actual ground space projector
$\Pi_{gs}$. It is usually characterized by two parameters $D,\Delta$
that upperbound the amount of entanglement it creates and its
closeness to the actual ground space projector. In
\cRefs{ref:ALV2012-AL, ref:Arad2013-1DAL} it was shown that when the
system has a unique ground state, the existence of an AGSP with
$D\cdot \Delta<1/2$ implies a bound of $O(\log D)$ on the
ground-state entanglement-entropy. Finding such ``good AGSP'', in
which $\log D$ scales like the boundary of $L$ is therefore
sufficient to prove an area-law. Our main result is a similar
condition for the maximally-mixed ground-state. To use it for
proving area-laws, we would need the following results about the
good AGSPs for 1D and 2D systems, which were used to show area-laws
in the unique ground state case.
\begin{fact}[A good 1D AGSP, Lemmas~4.1,~4.2 from \cRef{ref:Arad2013-1DAL}] 
\label{fact:1D-AGSP} 
  Let $H$ be a 2-body local Hamiltonian on a 1D lattice $\Lambda$
  with $d$-dimensional qudits and a spectral gap $\gamma>0$. Then
  for every bi-partition of the lattice into two contiguous regions
  $L\cup L^c$, there exist a family of $(D,\Delta)$-AGSPs
  $\{K(\ell,s)\}$ for integers $\ell,s$ such that
  \begin{align}
  \label{def:1D-AGSP}
    \Delta &= e^{-\bOmega{\ell \cdot(\gamma/s)^{1/2} }}, &
    D &= e^{\bigO{\log(d\ell)\cdot\max\{\ell/s, \sqrt{\ell}\}}} .
  \end{align}
\end{fact}

{~}

For the 2D case, we will use the following AGSP construction of
\cRef{ref:anshu2022-2DAL}
\begin{fact}[A good 2D AGSP, Theorem~4.4 from
  \cRef{ref:anshu2022-2DAL}] \label{fact:2D-AGSP} 
  
  Let $H$ be a frustration-free local Hamiltonian on a 2D
  rectangular lattice, defined over qudits of dimension
  $d=\bigO{1}$ and with a \emph{local} spectral gap
  $\gamma=\Omega(1)$. Then for every bi-partitioning of the system
  $L\cup L^c$ along a vertical cut of length $|\partial L|$ there
  exists a $(D,\Delta)$-AGSP with $D\cdot \sqrt{\Delta}<\frac{1}{2}$
  \footnote{In \cRef{ref:anshu2022-2DAL} the $\Delta$ parameter of
  the AGSP is defined by $\Delta=\norm{K-\Pi_{gs}}$, which
  translates to $\sqrt{\Delta}$ in our AGSP
  definition~\eqref{def:AGSP}. Therefore a $(D,\Delta)$-AGSP in
  \cRef{ref:anshu2022-2DAL} is actually a $(D,\sqrt{\Delta})$ in our
  convention.} and
  \begin{align}
  \label{def:2D-D}
    \log D = |\partial L|^{1 + O(\log^{-1/5}|\partial L|)} .
  \end{align}
\end{fact}

{~}

\subsection{Technical lemmas that are needed in the main proof}

We start with the following lemma about bounding the number of
distinct eigenvalues. 
\begin{lemma}[Spectrum discretization] 
\label{lem:discretiztion} Let $\eps>0$ be a small number. Let $\rho$
  be a sub-state and $\sigma=\sigma_L\otimes \sigma_R$ be a product
  state. Assume that $\rho\preceq t\sigma$ and $t\le d_L^2$, where
  $d_L$ is the Hilbert space dimension of subsystem $L$. There
  exists a state $\hat{\sigma}_L\in\mcD(\mcH_L)$ and a sub-state
  $\hat{\rho}$ such that $|\Spec(\hat{\sigma}_L)| \le
  7\log(d_L/\eps)$, $\norm{\hat{\rho}-\rho}_1 \le 2(\eps/d_L)^2$ and
  
  \begin{align*}
    \hat{\rho} \preceq 2 t
      \cdot(\hat{\sigma}_L\otimes \sigma_R).
  \end{align*}
  In addition, $\lambda_{max}(\hat{\rho})\le \lambda_{max}(\rho)$.
\end{lemma}

\begin{proof}
  We begin by projecting out the small
  eigenvalues of $\sigma_L$ from both sides of $\rho \preceq
  t\sigma$ and then discretize the spectrum of $\sigma_L$. Let
  $\Pi_L$ be the projector onto the eigenspace of $\sigma_L$ with
  eigenvalues greater than $\eps^4/d_L^7$, and set $\Pi \EqDef \Pi_L
  \otimes \Id_R$. Consider,
  \begin{align*}
    \Tr{(\Id-\Pi)\rho} & \leq  t\Tr{(\Id-\Pi)\sigma} 
    & \text{$(\rho \preceq t \sigma )$} \\
    &\leq d_L^2\Tr{(\Id-\Pi)\sigma} 
    & \text{$(t \le d_L^2)$} \\
    & = d_L^2\Tr{(\Id-\Pi_L)\sigma_L} 
    & \text{(definition of $\Pi$)} \\
    &\leq (\eps/d_L)^4. & \text{(definition of $\Pi_L$)} 
  \end{align*}
  Defining $\hat{\rho} \EqDef \Pi\rho\Pi$, we deduce from the gentle
  measurement lemma (Fact~\ref{fact:gentle}) that
  $\norm{\hat{\rho}-\rho}_1\le 2 (\eps/d_L)^2$. Additionally, this
  definition trivially satisfies the assertion
  $\lambda_{max}(\hat{\rho})\le \lambda_{max}(\rho)$.
  
  Next, consider the spectral decomposition of $\Pi_L\sigma_L\Pi_L$,
  \begin{align*}
    (\Pi_L\sigma_L\Pi_L) &=\sum_i \ell_i \ketbra{\ell_i}{\ell_i},
    \end{align*}
  Let $N$ be the smallest integer for which $2^N\cdot
  \frac{\eps^4}{d_L^7}\ge 1$. It is easy to see that $N
  \le 7\log\big(d_L/\eps\big)$. For every $n \in \{
  0,1,\ldots,N\}$ let $\lambda_n \EqDef 2^n\cdot
  \frac{\eps^4}{d_L^7}$ so that $\lambda_N\ge 1$, and set
  \begin{align*}
    \sigma'_L &\EqDef \sum_{n=1}^N
    \lambda_n \sum_{\ell_i\in (\lambda_{n-1},\lambda_n]} 
    \ketbra{\ell_i}{\ell_i}, &
  \hat{\sigma}_L \EqDef \sigma'_L/\tr(\sigma'_L) .
  \end{align*}
  That is, we define $\sigma_L’$ by rounding each eigenvalue to the nearest upper discretized 
  threshold $\lambda_n$ and then renormalize to get a valid state. Clearly, 
  $|\Spec(\hat{\sigma}_L)| = |\Spec(\sigma'_L)| = N\le
  7\log(d_L/\eps)$. In addition, for each $i:
  \ell_i\le \lambda_n\le 2\ell_i$ and therefore,
  \begin{align*}
    \Pi_L\sigma_L\Pi_L\preceq \sigma'_L 
    &\preceq 2\cdot \Pi_L \sigma_L \Pi_L,
  \end{align*}
  and so $\tr(\sigma'_L) \le 2$. Applying
  $\Pi=\Pi_L\otimes\Id_R$ on both sides of $\rho\preceq t\sigma $,
  we get
  \begin{align*} 
    \tilde{\rho} 
    &\preceq t\cdot(\Pi_L\sigma_L\Pi_L)\otimes(\sigma_R ) \\
    &\preceq t \cdot (\sigma'_L \otimes \sigma_R) \\
    &\preceq t \cdot 2\cdot(\hat{\sigma}_L\otimes\sigma_R).
  \end{align*}
\end{proof}

The following lemma, adapted from \cRef{ref:Anurag-Rahul-2016}, formalizes the brothers extension -- a key technical tool that we tailor to our framework.
\begin{lemma}[The brothers extension, adapted from 
  \cRef{ref:Anurag-Rahul-2016}] 
\label{lem:flat}
  Let $\delta >0$, $\rho_A= \sum_i a_i \ketbra{a_i}{a_i}$  be a sub-state expressed in its eigenbasis,
  and $\tau_A$ be a state
  on $\mcH_A$ and let $\log t\EqDef \dmax{\rho_A}{\tau_A}$ so that
  $\rho_A \preceq t\cdot \tau_A$. There exists an auxiliary Hilbert space
  $\mcH_B$ (the brothers space of dimension $d_B$), together with $\rho'_{AB}\in
  \mcD_-(\mcH_{A} \otimes \mcH_B)$ and a \emph{flat}
  $\sigma_{AB}\in\mcD(\mcH_{A}
  \otimes \mcH_B)$ obeying the bound 
  
  $$\norm{\rho_{A}
  - \rho'_{A}} _1 \le \norm{\rho_{AB}
  - \rho'_{AB}}_1 \le \delta,$$ and
  \begin{align*}
    \rho'_{AB} \preceq t\cdot \frac{32}{\delta^2}\cdot \sigma_{AB} 
      \qquad \Leftrightarrow \qquad
  \dmax{\rho'_{AB}}{\sigma_{AB}}
      \le \dmax{\rho_A}{\tau_A} + \log\big(32/\delta^2\big),
  \end{align*}
  where 
\begin{align*}
  P_\rho &\EqDef \sum_i \ketbra{a_i}{a_i} 
      \otimes \sum_{m=1}^{d_B a_i}\ketbra{m}{m}, & \quad \rho_{AB} \EqDef \frac{P_{\rho}}{d_B}.
  \end{align*}
  
  In addition we get: 
  \begin{enumerate}
    \item If $\tau_A$ is a product state $\tau_L\otimes
        \tau_R$, then $\sigma_{AB}$ is separable (with respect to
        the cut $L:RB$) with Schmidt rank at most $|\Spec(\tau_L)|$.
  
  \item The image of $\rho'_{AB}$ is contained in $\mcH_A \otimes \Span\{\ket 1,\ldots,\ket{d_B\lambda_{\max}(\rho_A)} \}_B$.
  Moreover, $\lambda_{\max}(\rho'_{AB})=1/d_B$.

  \item The statement remains valid if we change $d_B \mapsto d_B \cdot k$ for an integer $k\in\mathbb N$.
  \end{enumerate}

\end{lemma}

\begin{proof}
  Consider the spectral decomposition,
  \begin{align*}
    \rho_A &= \sum_i a_i \ketbra{a_i}{a_i}, &
    \tau_A &= \sum_j b_j \ketbra{b_j}{b_j}.
  \end{align*}
  Note that $\Span\{\ket{a_i}\}\subseteq \Span\{\ket{b_j}\}$. We
  assume that all $a_i$ and $t \cdot b_j$ are rational numbers (this can
  be assumed with an arbitrarily small perturbation of the states).
  Let $d_B$ be the smallest common multiple of their denominators\footnote{The choice of least common multiple is not crucial for the proof.
  Any common multiple will also work. Therefore we can choose any multiple of the $d_B$ and the proof will work the same.}.
  Let $\mcH_B$ be a Hilbert space with $\dim(\mcH_B)=d_B$. Suppose
  that $\delta>0$ is a small rational number and
  $p=\frac{\delta^2}{32\;t}$.  Define two projectors on $\mcH_A
  \otimes \mcH_B$ as follows,
  \begin{align}
  \label{def:projectors}
    P_\rho &\EqDef \sum_i \ketbra{a_i}{a_i} 
      \otimes \sum_{m=1}^{d_B a_i}\ketbra{m}{m}, &
    P_\sigma &\EqDef \sum_i \ketbra{b_i}{b_i} 
      \otimes\sum_{m=1}^{\min\{d_B b_i/p ,d_B\}} \ketbra{m}{m} .
  \end{align}
  Let $\rho_{AB} \EqDef \frac{P_{\rho}}{d_B}$ and $\sigma_{AB}
  \EqDef\frac{P_\sigma}{\Tr{\tilde P_\sigma}}$. Note that
  $\sigma_{AB}$ is a flat state. At this point it is easy to verify
  from the definition of $P_\sigma$ that if $\tau_A$ is a product
  state $\tau_A=\tau_L\otimes \tau_R$ then $\sigma_{AB}$ is
  separable (with respect to the cut $L:RB$) with Schmidt rank at
  most $|\Spec(\tau_L)|$.
  
  Next, we would like to upperbound $\rho_{AB}$ by $\sigma_{AB}$.
  Unfortunately, we do not have that
  $\supp(\rho_{AB})\subseteq\supp(\sigma_{AB})$. For this purpose,
  we truncate the projection of the small eigenvectors of $\sigma$
  from each $\ket{a_i}$. We start by introducing the following
  necessary fact from \cRef{ref:Anurag-Rahul-2016}:
  \begin{fact}[Claim 3.3 in \cRef{ref:Anurag-Rahul-2016}]
    \label{fact:Anurag}
    Let $p>0$. Then,
    \begin{align*}
        \forall i : \quad \sum_{j\;;\; b_j\le p a_i} |\braket{b_j|a_i}|^2 \le t\cdot p .
    \end{align*}
  \end{fact}
  For each $i$, let $\ket{a_i} = \sum_j \alpha_{ij} \ket{b_j}$ and
  define,
  \begin{align*} 
    \ket{\tilde a _i}  \EqDef \sum_{j\;;\; b_j\geq p a_i} \alpha_{ij} \ket{b_j} \quad ; \quad 
    \ket{\hat a _i}  \EqDef \frac{\ket{\tilde a _i}}{\norm{\ket{\tilde a _i}}}.
  \end{align*}
  That is, we project out from each eigenvector $\ket{a_i}$ the
  component of $\ket{b_j}$ with small eigenvalues and then
  renormalize the resulting state. From Fact~\ref{fact:Anurag} and
  choice of $p$ we get $|\braket{a_i|\hat a_i}|^2\geq
  1-\delta^2/16$. Define,
  \begin{align} 
    \label{eq:Prhotilde} 
    \tilde P_\rho \EqDef \sum_i \ketbra{\hat a_i}{\hat a_i} \otimes \sum_{m=1}^{d_B a_i}\ketbra m m  \quad ; \quad \tilde \rho_{AB} = \frac{\tilde P_\rho}{d_B}.
  \end{align}
  Note that indeed $\im(\tilde\rho_{AB})\subseteq\mcH_A \otimes \Span\{\ket 1,\ldots,\ket{d_B\cdot\lambda_{\max}(\rho_A)} \}_B$.
  We have the following claim. 
  \begin{claim} 
  \label{clm:brothers}
    \begin{enumerate}
        \item $\norm{\rho_{AB}-\tilde{\rho}_{AB}}_1\le \delta/2$.
        \item $\supp(\tilde\rho_{AB})\subseteq \supp(\sigma_{AB})$.
        \end{enumerate}
  \end{claim}
  
  \begin{proof}
    \begin{enumerate}
      \item Consider,
        \begin{align*}
          &\norm{\rho_{AB}-\tilde\rho_{AB}}_1 \\
          &\le \frac{1}{d_B}\sum_i\big\|\left(\ketbra{a_i}{a_i}
            - \ketbra{\tilde a_i}{\tilde a_i}\right)\otimes 
                \sum_{m=1}^{d_B a_i}\ketbra{m}{m}\big\|_1
                  & \text{(triangle inequality)}\\
          &= \frac{1}{d_B}\sum_i\big\|\left(\ketbra{a_i}{a_i}
            - \ketbra{\tilde a_i}{\tilde a_i}\right)\big\|_1
              \cdot \big\|\sum_{m=1}^{d_B a_i}\ketbra{m}{m}\big\|_1\\
          &= \sum_i a_i \norm{\left(\ketbra{a_i}{a_i} 
            - \ketbra{\tilde a_i}{\tilde a_i}\right)}_1 \\
          &= \sum_i a_i \cdot 2\sqrt{1-|\braket{a_i|\hat a_i}|^2} 
            & \text{(Fact~\ref{fact:norm})}\\
          & \le \sum_i a_i \cdot 2\sqrt{\delta^2/16} 
            & \text{($|\braket{a_i|\hat a _i}|^2 \ge 1-\delta^2/16$)}\\
          &\le \delta/2.
        \end{align*}
        The first equality follows from the multiplicativity of the
        trace norm under tensor products.  
      \item Let $\ket{\tilde a_i,m}$ be a basis element in the support 
        of $\tilde{P}_\rho$. Note in this case $m\le d_B a_i$.
        Recall that 
        \begin{align*}
          \ket{\tilde a_i,m } \propto \sum_{j\;;\; b_j \geq p a_i}  \alpha_{ij} \ket{b_j,m}.
        \end{align*}
        Note that each of the $\ket{b_j,m}$ in the summation above is
        in the support of $\tilde{P}_\sigma$ (from \Eq{eq:Prhotilde}).
        This is because $a_i \le \frac{b_j}{p}$ implies that
        $d_B\frac{b_j}{p} \geq d_Ba_i\geq m$.
      \end{enumerate}
  \end{proof}
  Let $\tilde\rho_{AB} = \sum_i c_i \ketbra{c_i}{c_i} $ be the
  spectral decomposition. From Fact~\ref{fact:bhatia} and
  Claim~\ref{clm:brothers},
  \begin{align} 
  \label{eq:eigenbound}
    \norm{\Eig^\downarrow(\tilde\rho_{AB}) - 
      \Eig^\downarrow(\rho_{AB})}_1 
      = \sum_i |c_i - 1/d_B| \le \norm{\tilde \rho_{AB}
        - \rho_{AB}}_1\le \delta/2.
  \end{align}
  Define,
  \begin{align*}
    \rho'_{AB} \EqDef \sum_{i} 
      \min\{c_i, \frac{1}{d_B}\} \cdot \ketbra{c_i}{c_i}.
  \end{align*}
  From \Eq{eq:eigenbound} and Claim~\ref{clm:brothers},
  \begin{align} \label{ineq:brothers-bound}
    \norm{\rho'_{A} - \rho_{A}}_1 
      \le \norm{\rho'_{AB} - \rho_{AB}}_1
      \le \norm{\rho'_{AB} - \tilde \rho_{AB}}_1
      + \norm{\tilde\rho _{AB} - \rho_{AB}}_1 \le \delta .
  \end{align}
  Since $\sigma_{AB}$ is a flat state and $\supp(\rho_{AB}')
  =\supp(\tilde\rho_{AB}) \subseteq \supp(\sigma_{AB})$, we use \Lem{lem:Dmax_flat_2} to write
  \begin{align} \label{eq:final-Dmax-brothers}
    \dmax{\rho'_{AB}}{\sigma_{AB}}
    = \log(\tr(P_\sigma)\cdot \lambda_{\max}(\rho'_{AB}))
    = \log(\tr(P_\sigma)\cdot 1/d_B).
  \end{align}
  We finish by noting that $\tr(P_\sigma)\le \frac{32 d_B
  t}{\delta^2} $.  
\end{proof}

{~}

The following are two results about bounded Schmidt-rank operators. 
\begin{lemma} 
\label{lem:SR-Dmax}
  Let $\rho_{LR} = \sum_{i=1}^D p_i\rho^i_L\otimes\rho^i_R$ be a state,
  where $\{p_i\}$ is a probability distribution and 
  $\{\rho^i_L\}_i, \{\rho^i_R\}_i$ are states.
  Then there exist  states $\theta_L, \theta_R$ such that
  \begin{align*}
    \rho_{LR} \preceq D^2\cdot \theta_L\otimes\theta_R.
  \end{align*}
\end{lemma}

\begin{proof} \ \\
  Define $\theta_L \EqDef \frac{1}{D}\sum_{i=1}^D \rho_L^i$ and
  $\theta_R \EqDef \frac{1}{D}\sum_{j=1}^D \rho_R^j$. Then,
  \begin{align*}
    \theta_L\otimes\theta_R 
          = \frac{1}{D^2}\sum_{i,j}\rho_L^i\otimes\rho_R^j
          \succeq \frac{1}{D^2}\sum_i\rho_L^i\otimes\rho_R^i
          \succeq \frac{1}{D^2}\sum_ip_i\rho_L^i\otimes\rho_R^i
          = \frac{1}{D^2}\rho_{LR} .
  \end{align*}  
\end{proof}

\begin{lemma} 
\label{lem:boundSR}
  Let $ M = \sum_{i=1}^D \alpha_i (L_i \otimes R_i)$ where
  $\|M\|_2=1$. There exists states $\tau_L, \tau_R$ such that,
  \begin{align*}
    M M^\dagger \preceq D^2 (\tau_L \otimes \tau_R).
  \end{align*}
  In other words, 
  \begin{align*}
    \maxinf{L}{R}_{MM^\dagger} \le 2  \cdot \log \SR{L}{R}_{M}.
  \end{align*}
\end{lemma}
\begin{proof}
  Let,
  \begin{align*}
    \ket{v}_{L \tilde{L} R \tilde{R}} = \myvec{M} = \sum_{i=1}^D \alpha_i \cdot ( \myvec{L_i}_{L\tilde{L}} \otimes \myvec{R_i}_{R\tilde{R}}).
  \end{align*}
  From Fact~\ref{fact:vec}, $\Ptr{\ketbra{v}{v}}{\tilde{L}
  \tilde{R}} = M M^\dagger$. Let $\Pi_{L \tilde{L}}$ be the
  projector onto $\mathsf{span}\{\myvec{L_i}\}$ and $\tau_{L
  \tilde{L}} = \frac{\Pi_{L \tilde{L}}}{\Tr{\Pi_{L \tilde{L}}}}$.
  Similarly let $\Pi_{R \tilde{R}}$ be the projector onto 
  $\mathsf{span}\{\myvec{R_i}\}$ and $\tau_{R \tilde{R}} =
  \frac{\Pi_{R \tilde{R}}}{\Tr{\Pi_{R \tilde{R}}}}.$ Note
  $\Tr{\Pi_{L \tilde{L}}} \le D$ and $\Tr{\Pi_{R \tilde{R}}} \le D$.
  Consider,
  \begin{align*}
    \ketbra{v}{v} &\preceq \Pi_{L \tilde{L}} \otimes \Pi_{R \tilde{R}} \preceq D^2 (\tau_{L \tilde{L}} \otimes \tau_{R \tilde{R}}) ,\\
    \Rightarrow M M^\dagger &\preceq D^2(\tau_{L} \otimes \tau_{R}). & \mbox{(monotonicity of partial trace)}
  \end{align*} 
\end{proof}

\section{Proof of the main results} 
\label{sec:proofs}

In this section we present the proof of our main area-law
bootstrapping result, \Thm{thm:bootstrapping}, as well as the
proofs of Corollaries~\ref{corol:mutual-info},\ref{corol:1D-AL},
\ref{corol:2D-AL}.

\subsection{Proof of \Thm{thm:bootstrapping}}
\label{sec:main-AL-proof}

As in the overview of the proof, we slightly change the notation and
denote the bi-partition of the lattice by $L\cup R$, instead of
$L\cup L^c$. We let $d_L$ denote the dimension of the Hilbert space
of subsystem $L$, e.g. $d_L= d^{|L|}$ for $d$-dimensional qudits.

Given a $\eps>0$, our goal is to find a state $\rho'$ such that
\begin{align*}
    \norm{\rho'-\rho}_1\le \eps \qquad \text{and} \qquad \maxinf{L}{R}_{\rho'}\le 2\log D
+ 12\log \left(\frac{\log d_L}{\eps}\right) + O(1).
\end{align*} Our strategy is
to construct a sequence of (sub-)states $\gs=\rho^{(0)} \to
\rho^{(1)} \to \rho^{(2)} \to \ldots$, together with corresponding
product states $\tau^{(k)}=\tau_L^{(k)}\otimes\tau^{(k)}_R$ and
bounds $t^{(k)}$ such that $\rho^{(k)} \preceq
t^{(k)}\cdot\tau^{(k)}_L\otimes\tau^{(k)}_R$. This implies that
$\maxinf{L}{R}_{\rho^{(k)}}\le \log(t^{(k)})$. On a very high level,
every $\rho^{(k)}, \tau^{(k)}$ are obtained from $\rho^{(k-1)},
\tau^{(k-1)}$ by first ``discretizing'' and truncating their
eigenvalues, and then applying an AGSP. Our construction guarantees
that consecutive $\rho^{(k)}$ are close to each other, and are
therefore close to $\gs$. If all the $t^{(k)}$ are decreasing
rapidly enough, then at some point we will get a $\rho^{(k)}$ with
sufficiently low $\maxinf{L}{R}_{\rho^{(k)}}$, which, in turn will
imply a bound on $\maxinfeps{L}{L^c}_{\gs}$. On the other hand, if
not all the $t^{(k)}$ are decreasing rapidly, then for some $k$ it
must be that $t^{(k+1)}\ge t^{(k)}/2$. This condition, together with
the fact that the states $\rho^{(k+1)}$, $\tau^{(k+1)}$ are obtained
from $\rho^{(k)}$, $\tau^{(k)}$ using a ``good AGSP'' will enable us
to get an upper bound on $t^{(k)}$ --- which will yet again imply an
upper bound on $\maxinfeps{L}{L^c}_{\gs}$.

We begin with the definition of the sequence of states
$\{\rho^{(k)}\}$ and $\{\tau^{(k)}\}$, which are defined by
induction. For $k=0$, we define $\rho^{(0)}\EqDef \gs$, and let
$\tau^{(0)}=\tau^{(0)}_L\otimes\tau^{(0)}_R$ be a product state
such that $\gs \preceq 2^{\maxinf{L}{R}_\gs}\cdot \tau^{(0)}$.
Setting $t^{(0)} \EqDef 2^{\maxinf{L}{R}_\gs}$, we obtain
\begin{align*}
  \rho^{(0)} \preceq t^{(0)} \cdot \tau^{(0)} .
\end{align*}

Let us now define $\rho^{(k+1)}$, $\tau^{(k+1)}$ from
$\rho^{(k)}$, $\tau^{(k)}$. For brevity, we write $\rho=\srho{k}$,
$t=\st{k}$ and $\tau = \stau k$. Our construction consists of 4
steps.

{~}

\noindent\textbf{Step I: Discretization: $\rho^{(k)} \to
\hat{\rho}$, $\tau^{(k)} \to \hat{\tau}$} \ \\

We begin by defining a small parameter
\begin{align}
\label{def:delta}
  \delta \EqDef \Big(\frac{\eps}{50 \log (d_L)}\Big)^2 
\end{align}  
and applying \Lem{lem:discretiztion} on $\rho \preceq t
(\tau_L\otimes\tau_R)$ with $\epsilon ' = \epsilon/50$. This
produces a new sub-state $\hat \rho$ and a product state $\hat{\tau}
= \hat{\tau}_L\otimes \tau_R$ such that\footnote{The assumption $\st
0 \le d_L^2$ is promised by Fact~\ref{lem:Imax_bound}. Also for
$k\geq 1: \st k \le \st 0$ as will be later shown.} 
\begin{align} \label{eq:mixed-first-ineq}
    \hat{\rho} &\preceq 2 t \hat{\tau}
\end{align}
with 
\begin{align*}
  |\Spec(\hat{\tau}_L)|  = 7\log(50 d_L/\eps) 
  < (50 \log(d_L)/\eps)^2 = \frac{1}{\delta}
\end{align*}
and 
\begin{align}
\label{eq:distance-to-hat}
  \norm{\hat{\rho} -\rho}_1 \le 2\big(\frac{\eps}{50 d_L}\big)^2
    < \delta.
\end{align}

{~}

\noindent\textbf{Step II: Brothers extension: $\hat{\rho}
\to\rho_{AB}'$, $\hat{\tau}\to \sigma_{AB}$} \ \\
  
Let $\mcH_A = \mcH_L \otimes \mcH_R$. We now use a mapping known
as the `brothers extension', in which we introduce an auxiliary
Hilbert space $\mcH_B$, known as the `brothers space', and extend
$\hat{\rho} \to \rho_{AB}'$ and $\hat{\tau}\to \sigma_{AB}$. The
brothers extension should be viewed as a purely mathematical
tool, without any direct physical meaning. Its purpose is
transforming the product state $\hat{\tau}$ to a {flat state}
$\sigma_{AB}$, which, in turn will enable us to relate its min
and max entropies. The brothers extension is done by invoking
\Lem{lem:flat} with the parameter $\delta$ and $\rho_A =
\hat{\rho}, \tau_A = \hat{\tau}$. Recalling that
$\hat{\rho}\preceq 2t\hat{\tau}$, we obtain a sub-state
$\rho'_{AB}$ and a flat state $\sigma_{AB}$ such that 
\begin{align} 
\label{eq:mixed-second-ineq}
  \rho'_{AB} &\preceq 2t\cdot
    (32/\delta^2)\cdot\sigma_{AB} \\
 \label{eq:distance-to-prime}
  \norm{\rho'_A-\hat{\rho}}_1 &\le \delta, \\
   \SR{L}{RB}_{\sigma} &\le |\Spec(\hat{\tau}_L)| \le 1/\delta.
\end{align}
Defining $f(\delta)\EqDef 64/\delta^2$,
\Ineq{eq:mixed-second-ineq} implies
\begin{align}
\label{eq:f-ineq}
  \rho'_{AB} &\preceq t\cdot f(\delta)\cdot\sigma_{AB} .
\end{align}

{~}

\noindent\textbf{Step III: Applying the AGSP: $\rho'_{AB}\to
\tilde{\rho}_{AB}$, $\sigma_{AB}\to \theta_L\otimes\theta_{RB}$} \ \\

The next step would be to apply our $(D,\Delta)$-AGSP on both
sides of the above inequality. However, our $(D,\Delta)$-AGSP $K$
acts on $\mcH_A$, while the operators act in the extended space
$\mcH_A\otimes\mcH_B$. We therefore extend $K$ to act on
$\mcH_A\otimes\mcH_B$: we let $r \EqDef
\dim(V_{gs})$ (i.e., $r$ is the ground space degeneracy) and
then define 
\begin{align}     \label{eq:extended-AGSP}
  \Pi_r\EqDef \sum_{m=1}^{d_B/r}\ketbra m m \quad , \quad
    \Pi'_{gs}\EqDef \Pi_{gs}\otimes \Pi_r \quad , \quad 
    \gs_{AB}\EqDef \frac{\Pi'_{gs}}{\tr(\Pi'_{gs})} 
      \quad , \quad  K_{AB} \EqDef K \otimes \Pi_r .
\end{align}
Note that $K_{AB}$ is a $(D,\Delta)$-AGSP for (the extended ground
space) $\Pi_{gs}'$ and $\gs = \tr_B(\gs_{AB})$. Applying
$K_{AB}$ on both sides of \Ineq{eq:f-ineq}, we get
\begin{align*}
    \tilde \rho_{AB} \EqDef K_{AB}\rho'_{AB}K_{AB}^\dagger
    \preceq t \cdot f(\delta) \cdot  K_{AB}  \sigma_{AB} K_{AB}^\dagger.
\end{align*}
As $\sigma_{AB}$ is a flat state, it follows that
$\sqrt{\sigma_{AB}} \propto \sigma_{AB}$ and therefore
\begin{align*}
  \SR{L}{RB}_{\sqrt{\sigma_{AB}}} =\SR{L}{RB}_{\sigma_{AB}}
    \le 1/\delta.
\end{align*}
Moreover, since $\SR{L}{RB}_{K_{AB}}\le D$, we get
$\SR{L}{RB}_{K_{AB}\sqrt{\sigma_{AB}}} \le D/\delta$. Invoking
\Lem{lem:boundSR} with $M=K_{AB}\sqrt{\sigma_{AB}}$, we find that
there exists a product state $\theta_L\otimes\theta_{RB}$ such
that
\begin{align} \label{eq:mixed-fourth-ineq}
    \tilde \rho_{AB}  \preceq t  \cdot f(\delta)\cdot \delta^{-2}
      \cdot D^2 \cdot\Tr{K_{AB}  \sigma_{AB} K_{AB}^\dagger} 
      \cdot (\theta_L\otimes \theta_{RB}).
\end{align}

{~}

\noindent\textbf{Step IV: Tracing out and truncation:
$\tilde{\rho}_{AB}\to \tilde{\rho}_A \to \rho^{(k+1)}$, 
$\theta_L\otimes\theta_{RB} \to \theta_L\otimes\theta_{R} 
=\tau^{(k+1)}$}  \\

Once we applied the AGSP on the extended space, we return to the
original space $\mcH_A=\mcH_L\otimes \mcH_R$ by tracing out the
brothers space:
\begin{align*}
  \tilde{\rho}_A \EqDef \Ptr{\tilde{\rho}_{AB}}{B} .
\end{align*}
The final step is to round each eigenvalue of $\tilde{\rho}_A$
which is larger than $\lambda_{\max}(\srho{k})$. Formally, let 
$\tilde{\rho}_A = \sum_i \lambda_i \ketbra{\psi_i}{\psi_i}$
be the spectral decomposition of $\tilde{\rho}_A$. Then we define
\begin{align}
  \label{def:next-rho}
	\rho^{(k+1)} &\EqDef \sum_i \lambda'_i \ketbra{\psi_i}{\psi_i}, &
	\lambda'_i &\EqDef \min(\lambda_i, \lambda_{\max}(\rho^{(k)}) .
\end{align}
This ensures 
\begin{align}
\label{eq:lambdamax} 
  \lambda_{\max}(\srho{k+1}) \le \lambda_{\max}(\srho{k}).  
\end{align}
In addition, we define $\tau^{(k+1)}=\tau^{(k+1)}_L\otimes
\tau^{(k+1)}_R$ by
\begin{align}
  \label{eq:next-tau}
	\tau^{(k+1)}_L &\EqDef \theta_L, &
	\tau^{(k+1)}_R &\EqDef \theta_R = \tr_B\theta_{RB} .
\end{align}
By definition, $\rho^{(k+1)} \preceq \tilde{\rho}_A$ and so by
tracing out the brothers space in \eqref{eq:mixed-fourth-ineq}, we
obtain
\begin{align} \label{eq:mixed-fourth-ineq-A}
  \rho^{(k+1)}  \preceq t\cdot f(\delta)\cdot \delta^{-2}
      \cdot D^2 \cdot\Tr{K_{AB}  \sigma_{AB} K_{AB}^\dagger} 
      \cdot \tau^{(k+1)}_L\otimes\tau^{(k+1)}_R.
\end{align}
To define $t^{(k+1)}$ we will use the following claim, whose proof
we defer to later. 
\begin{claim} \label{clm:quantities} \ 
  \begin{enumerate}

    \item \label{clm:quantities-1} 
      \begin{align*}
        \Ptr{K_{AB}\, \rho'_{AB}\, K_{AB}^\dagger}{B} 
        = K_A \rho'_A K^\dagger_A.
      \end{align*}
       
    \item \label{clm:quantities-2} For every $k$, assuming 
    $\Delta\le \delta$,
      \begin{align*}
        \norm{\rho^{(k+1)}-\rho^{(k)}}_1\le 20 \sqrt\delta.
      \end{align*}

    \item \label{clm:quantities-3}  
      \begin{align*}
         \tr\Big(K_{AB}\,\sigma_{AB}\, K_{AB}^\dagger\Big)
         \le \Delta + \frac{1}{\delta^2} 
         \cdot 2^{-{\maxinf{L}{R}}_{\rho_A'}} .
      \end{align*}
      
  \end{enumerate}
\end{claim}

Using Bullet~\ref{clm:quantities-3} of the claim, 
\Ineq{eq:mixed-fourth-ineq-A} becomes
\begin{align*} 
  \rho^{(k+1)}  &\preceq t \cdot f(\delta)\cdot 
     \delta^{-2}\cdot D^2\cdot 
         \left( \Delta + \frac{1}{\delta^2}\cdot 
        2^{-\maxinf{L}{R}_{\rho_A'}}\right) \cdot
     \tau^{(k+1)}_L\otimes\tau^{(k+1)}_R.
\end{align*}
We now use our main structural
assumption on the AGSP, namely, $D^2\cdot \Delta \le c_0(\eps/\log
d_L)^8$, and choose 
\begin{align}
  c_0 \EqDef 10^{-16}.
\end{align}
Using the definition of $\delta$ in \eqref{def:delta} and
the definition of $f(\delta)=64/\delta^2$, it is easy to verify that
such $c_0$ guarantees that $f(\delta)\cdot \delta^{-2}\cdot
D^2\cdot\Delta\le 1/4$ and therefore,
\begin{align*}
  \rho^{(k+1)} &\preceq t\left(
    \frac{1}{4}+ 2^{-\maxinf{L}{R}_{\rho'_A}}
      \cdot f(\delta)\cdot D^2/\delta^4\right)
        \cdot\tau^{(k+1)}_L\otimes\tau^{(k+1)}_R.
\end{align*}
Recalling that $t=t^{(k)}$, we define
\begin{align}
\label{eq:next-t}
  t^{(k+1)} &\EqDef t^{(k)}\cdot 
    \left(\frac{1}{4} + 2^{-\maxinf{L}{R}_{\rho_A'}}
      \cdot f(\delta)\cdot D^2/\delta^4\right),  
\end{align}
and obtain $\rho^{(k+1)} \preceq t^{(k+1)}\cdot \tau^{(k+1)}$ as
required.

{~}  

Now that we have defined our sequence $\rho^{(k)} \preceq
t^{(k)}\cdot \tau^{(k)}_L\otimes\tau^{(k)}_R$, let us understand why
it implies a bound on $\maxinfeps{L}{R}_\gs$.  We first observe that
there must be an integer $k\le 2\log d_L$ such that $t^{(k+1)} \ge
t^{(k)}/2$. Otherwise, for $\ell=\lceil 2\log d_L\rceil$,
\begin{align*}
    t^{(\ell)} < \frac{t^{(\ell-1)}}{2}
    < \frac{t^{(\ell-2)}}{2^2}
    < \cdots 
    < \frac{t^{(0)}}{2^\ell} .
\end{align*}
But since $t^{(0)} \le d_L^2$ (Fact~\ref{lem:Imax_bound}), we get
that $t^{(\ell)}<1$, which is a contradiction. 

Let us then take $k\le 2\log d_L$ to be an integer for which
$t^{(k+1)} \geq t^{(k)}/2$. From the definition of $t^{(k+1)}$, we
get
\begin{align*}
  \frac{t^{(k)}}{2} & \le  \st{k}  \cdot \left(\frac{1}{4} 
    + 2^{-\maxinf{L}{R}_{\rho_A'}} 
      \cdot f(\delta)\cdot D^2/\delta^4 \right).
\end{align*}
Dividing both sides by $t^{(k)}$ and re-grouping the terms, we get
\begin{align*}
  2^{\maxinf{L}{R}_{\rho_A'}} 
    &\le 4D^2 \cdot f(\delta)/\delta^4 
    = 256D^2/\delta^6
      = 256D^2\cdot\left(\frac{50\log d_L}{\eps}\right)^{12}.
\end{align*}
Then taking $\log$ on both sides shows that
\begin{align} 
\label{eq:final_I_max}
  \maxinf{L}{R}_{\rho_A'}  = 2\log D + 12\log(\log d_L/\eps) + O(1).
\end{align}

Finally, we need to show that $\norm{\rho_A'-\gs}_1\le \eps$.  By
Claim~\ref{clm:quantities} Bullet \ref{clm:quantities-2}, we get
that for every $\ell=0, 1, \ldots, k \le 2\log(d_L)$
\begin{align*}
  \big\|\rho^{(\ell+1)}-\rho^{(\ell)}\big\|_1\le 20\sqrt \delta ,
\end{align*}
and therefore by a telescopic argument,
\begin{align*}
  \big\|\rho^{(k)}-\gs\big\|_1 
= \big\|\rho^{(k)}-\rho^{(0)}\big\|_1
  \le 20\cdot k \cdot \sqrt{\delta}
  \le 20\cdot 2\log d_L\cdot \frac{\eps}{50\log d_L} \le \frac{4}{5}\eps .
\end{align*}
In our notation, $\rho^{(k)}=\rho$, and so by inequalities 
\eqref{eq:distance-to-hat} and \eqref{eq:distance-to-prime} we
get
\begin{align*}
  \norm{\rho'_A-\rho^{(k)}}_1 
  \le \norm{\rho'_A -\hat{\rho}}_1 
     + \norm{\hat{\rho} -\rho^{(k)} }_1
  \le \delta + \delta = 2\delta ,
\end{align*}
which brings us to 
\begin{align*}
  \norm{\rho'_A - \gs}_1 \le \frac{4}{5}\eps + 2\delta =
    \frac{4}{5}\eps + 2\Big(\frac{\eps}{50\log d_L}\Big)^2 \le \eps.
\end{align*}

{~}  

We finish the proof by proving Claim~\ref{clm:quantities}.

\begin{proof}[ of Claim~\ref{clm:quantities}] 
  For brevity denote $\rho = \rho^{(k)}, \eta =\rho^{(k+1)}$.
  \begin{enumerate}
    \item By inequality \eqref{eq:lambdamax}, we get
      $\lambda_{\max}(\rho^{(k)})\le \lambda_{\max}(\rho^{(0)}) =
      1/r$, where $r$ is the degeneracy of the ground space. In
      addition, as promised by \Lem{lem:discretiztion}, moving
      from $\rho$ to $\hat{\rho}$ does not increase the largest
      eigenvalue of $\hat{\rho}$ and so
      $\lambda_{\max}(\hat{\rho}) \le 1/r$. As promised from
      \Lem{lem:flat}, 
      \begin{align*}
        \image(\rho'_{AB}) \subseteq \mcH_A \otimes \Span(\ket 1,\ldots , \ket{d_B\cdot\lambda_{\max}(\rho)})
        \subseteq \mcH_A \otimes \Span(\ket 1,\ldots , \ket{d_B/r}).
      \end{align*}
      That is the image of $\rho_{AB}'$ in system $B$ is
      completely contained in the image of $\Pi_r$. As a result
      $K_{AB}=K\otimes \Pi_r$ acts as identity on the $B$ part;
      that is,
      \begin{align*}
         K_{AB}\rho _{AB}' K_{AB}^\dagger =
            (K\otimes \Id_B)\rho _{AB}' (K^\dagger \otimes \Id_B),
      \end{align*}
      and Bullet~\ref{clm:quantities-1} is achieved.
        
    \item Recall, to get from $\rho_A$ to $\eta_A$ we perform 
      the following steps: 
      \begin{enumerate}
        \item Obtain $\hat{\rho}$ from $\rho$ using 
          \Lem{lem:discretiztion}.
    
        \item Obtain $\rho'_{AB}$ using \Lem{lem:flat}, and
          then $\rho_A'=\tr_{B}\rho'_{AB}$.
    
        \item Obtain $\tilde{\rho}_A =
          \tr_B(K_{AB}\rho'_{AB}K^\dagger_{AB}) = K\rho'_A K^\dagger$.
    
        \item Truncate the eigenvalues of $\tilde \rho_A$ which 
          exceed $\lambda_{\max}(\rho_A)$ to
          $\lambda_{\max}(\rho_A)$.
      \end{enumerate}
      We upper bound the trace distance introduced by each of
      these steps.  
      \begin{enumerate}
        \item \Lem{lem:discretiztion} promises that 
          $\norm{\rho-\hat{\rho}}_1\le \delta$.
    
        \item  \Lem{lem:flat} ensures that $\norm{\hat{\rho} 
          -\rho'_A}_1\le \delta$.
        
        \item We first show that $\Tr{(\Id_A-\Pi_{gs})\rho}\le \Delta$.
          Recall that $\rho=\rho^{(k)}$ and therefore
          \begin{align*}
            \rho\preceq \tilde{\rho}_A^{(k-1)} 
            = \tr_B \tilde{\rho}_{AB}^{(k-1)}
            = \tr_B K_{AB}(\rho'_{AB})^{(k-1)} K_{AB}^\dagger.
          \end{align*}
          Then by Bullet~\ref{clm:quantities-1} $\tr_B\Big(
          K_{AB}\cdot(\rho'_{AB})^{(k-1)}\cdot K_{AB}^\dagger\Big)
          = K\cdot(\rho'_A)^{(k-1)}\cdot K^\dagger$ and therefore
          $\rho\preceq K(\rho'_A)^{(k-1)}K^\dagger$. From this
          inequality, we conclude
          \begin{align*}
            \Tr{(\Id_A-\Pi_{gs})\rho} \le 
              \Tr{(\Id_A-\Pi_{gs})K(\rho_A')^{(k-1)}K^\dagger}\le \Delta,
          \end{align*}
          where we used properties of the AGSP $K$ from
          Definition~\ref{def:AGSP}. Thus,
          \begin{align*}
            \Tr{(\Id_A-\Pi_{gs})\rho_A'}\le \Tr{(\Id_A-\Pi_{gs})\rho} 
            + \norm{\rho_A-\rho'_A}_1\le \Delta+2\delta \le 4\delta .
          \end{align*}
          In the last inequality, we used the fact that $\Delta\le
          \delta$. This can be seen from the AGSP condition
          $D^2\cdot \Delta \le c_0\big(\eps/|L|\big)^8$, together
          with the definition of $\delta$ in \eqref{def:delta} and
          our choice of $c_0=10^{-16}$. Using the gentle
          measurement lemma (Fact~\ref{fact:gentle}) we deduce,
          \begin{align*}
            \norm{\rho'_A- \Pi_{gs}\rho'_A \Pi_{gs}}_1 
              \le 4\sqrt{\delta}.
          \end{align*}
			
          Next, using the fact that $K$ commutes with $\Pi_{gs}$,
          together with the fact 
          \begin{align*}
              \tilde{\rho}_A =
          \tr_B(K_{AB}\rho'_{AB}K^\dagger_{AB}) = K\rho'_A
          K^\dagger,
          \end{align*}
          we deduce that
          $\Pi_{gs}\tilde{\rho}_A\Pi_{gs} =
          \Pi_{gs}\rho'_A\Pi_{gs}$. Therefore,
          \begin{align*}
            \norm{\tilde \rho_A-\rho_A'}_1 &
            \le \norm{\tilde{\rho}_A-\Pi_{gs}\tilde{\rho}_A\Pi_{gs}}_1 
              + \norm{\rho_A'-\Pi_{gs}\rho'_A\Pi_{gs}}_1 
                 & \text{(triangle inequality)}\\
            &= \norm{K(\rho_A'-\Pi_{gs}\rho'_A\Pi_{gs})K^\dagger}_1 
              +\norm{\rho_A'-\Pi_{gs}\rho'_A\Pi_{gs}}_1\\
            & \le  4\sqrt{\delta} + 4\sqrt{\delta} 
              & \text{($K$ is a contractive map)}\\
            &= 8\sqrt{\delta},
          \end{align*}
          where in the second inequality we used the fact
          that\footnote{This follows from the fact that $K$ fixes
          ground states and shrinks the orthogonal part, i.e.
          decomposing $\ket\psi=
          \ket{\psi_{gs}}+\ket{\psi_{gs}^\bot}$, we get
          $\norm{K\ket \psi}^2=\norm{\ket{\psi_{gs}}}^2 +
          \norm{K\ket{\psi_{gs}^\bot}}^2\leq \norm{\ket
          \psi}^2\leq \norm{\ket{\psi_{gs}}}^2 +
          \norm{\ket{\psi_{gs}^\bot}}^2=\norm{\ket \psi}^2$ for
          any $\ket \psi$.} $\norm{K}=\norm{K^\dagger}=1$, hence
          using Holder inequality we have for any operator
          $O\in\mcL(\mcH)$, $\norm{KOK^\dagger}_1\leq
          \norm{O}_1$. 
          
        \item Combining the previous, we get,
          \begin{align*}
            \norm{\tilde{\rho}_A -\rho_A}_1 
              \le 8\sqrt{\delta} + 2\delta
              \le 10\sqrt{\delta} .
          \end{align*}
          From Fact~\ref{fact:bhatia}, we get
          \begin{align}  \label{corr:Bhatia}
          \begin{split}
            \norm{\Eig^\downarrow(\tilde{\rho}_A) 
               -\Eig^\downarrow(\rho)}_1 & 
               = \sum_{i\;;\;\lambda_i^\downarrow(\tilde \rho_A)\geq 
               \lambda_i^\downarrow(\rho)}
               \big(\lambda_i^\downarrow(\tilde \rho_A) 
                 - \lambda_i^\downarrow( \rho)\big)\\
               & + \sum_{i\;;\;\lambda_i^\downarrow(\tilde \rho_A)< 
               \lambda_i^\downarrow(\rho)}
               \big(\lambda_i^\downarrow(\rho)-\lambda_i^\downarrow(\tilde \rho_A)\big)\\
              & \le \norm{\tilde{\rho}_A -\rho}_1 
               \le 10\sqrt\delta .
               \end{split}
          \end{align}
          Recall that $\rho^{(k+1)}$ was obtained from $\tilde
          \rho_A$ by rounding down each eigenvalue
          $\lambda_i^\downarrow(\tilde \rho_A)$ which is larger
          than $\lambda^\downarrow_0(\rho)$ to
          $\lambda^\downarrow_0(\rho)$. Therefore
          \begin{align*}
            \norm{\rho^{(k+1)}-\tilde{\rho}_A}_1 
            = \sum_{\lambda_i^\downarrow(\tilde \rho_A)\geq 
              \lambda_0^\downarrow(\rho)}\!\!\!\!
              \big[ \lambda_i^\downarrow(\tilde \rho_A)
                - \lambda_0^\downarrow(\rho)\big].
          \end{align*}
          Due to the fact that any $\lambda_i^\downarrow(\tilde
          \rho_A)$ is lesser or equal than
          $\lambda_0^\downarrow(\tilde \rho_A)$, the expression
          above is necessarily smaller than the first sum written
          in \Eq{corr:Bhatia}, and therefore
          \begin{align*}
            \norm{\rho^{(k+1)}-\tilde{\rho}_A}_1 
              \leq  \norm{\Eig^\downarrow(\tilde{\rho}_A) 
               -\Eig^\downarrow(\rho)}_1
                &\le 10\sqrt\delta .
          \end{align*}
        \end{enumerate}
        Combining,
        \begin{align*}
            \norm{\rho^{(k)}-\rho^{(k+1)}}_1 \le 
            \norm{\rho^{(k)} - \tilde{\rho}_A}_1 +
            \norm{\tilde{\rho}_A- \rho^{(k+1)}}_1
            \le 20\sqrt{\delta}.
        \end{align*}
        
      \item Consider,
        \begin{align*}
            \tr(K_{AB}\sigma_{AB}K_{AB}^\dagger)
            &= \tr(K_{AB}^\dagger K_{AB}\sigma_{AB})\\
            &= \Tr{K_{AB}^\dagger \Pi_{gs}' K_{AB} \sigma_{AB} }
            + \Tr{K_{AB}^\dagger(\Id-\Pi_{gs}') K_{AB} \sigma_{AB} }\\
            &= \tr(\Pi_{gs}'\sigma_{AB})
             +\Tr{K_{AB}^\dagger(\Id-\Pi_{gs}')K_{AB}\sigma_{AB}}\\
            &\le \tr(\Pi_{gs}'\sigma_{AB}) + \Delta,
        \end{align*}
        where in the third equality we used 
        $K_{AB}^\dagger \Pi_{gs}' K_{AB}=\Pi_{gs}'$
        and in the last inequality we used the fact that $K_{AB}$
        is a $(D,\Delta)$-AGSP and so 
        $K_{AB}(\Id-\Pi_{gs}) K_{AB}^\dagger 
        \preceq \Delta(\Id -\Pi_{gs})$. To upperbound
        $\tr(\Pi_{gs}'\sigma_{AB})$, we use the fact that
        $\sigma_{AB}$ is a flat state and therefore
        \begin{align*}
          \Tr{\Pi_{gs}'\cdot \sigma_{AB}} & =
          \frac{1}{d_{\sigma}}\Tr{\Pi_{gs}'\cdot \Pi_{\sigma}} \\ &
          \le \frac{1}{d_{\sigma}} \Tr{\Pi_{gs}'} \\ & =
          \frac{d_B}{d_{\sigma}}\\ &=
          2^{-\dmax{\rho'_{AB}}{\sigma_{AB}}}. &
          \mbox{(\Eq{eq:final-Dmax-brothers})}
        \end{align*}
        Finally, monotonicity of $\mathrm{D}_{\max}$ under partial
        trace gives
        \begin{align*}
          \frac{1}{2^{\dmax{\rho'_{AB}}{\sigma_{AB}}}} 
            &\le \frac{1}{2^{\dmax{\rho'_A}{\sigma_A}}} \\
            &\le \frac{1}{\delta^2 \cdot 
              2^{\maxinf{L}{R}_{\rho^{\prime}_A}}}.
        \end{align*}

        The last inequality follows from the following arguments.
        From \Lem{lem:flat} part 1 we get that 
        $\sigma_{AB}$ and hence $\sigma_A$ is separable with 
        $\SR{L}{R}_{\sigma_A} \le 1/\delta$. This allows us to invoke
        \Lem{lem:SR-Dmax} to show that 
        \begin{align*}
         \rho'_A \preceq 2^{\dmax{\rho'_{A}}{\sigma_{A}}} 
           \sigma_A \preceq \frac{2^{\dmax{\rho'_{A}}
             {\sigma_{A}}}}{\delta^2}\theta,
        \end{align*}
        for some product state $\theta=\theta_L\otimes\theta_R$.
        Therefore, $2^{\maxinf{L}{R}_{\rho'_A}}\leq
        \frac{2^{\dmax{\rho'_{A}}{\sigma_{A}}}}{\delta^2}$, which
        proves the inequality.  This completes the proof.
    \end{enumerate}
\end{proof}

\vspace{2mm}

\subsection{Proof of Corollary~\ref{corol:mutual-info}
(bootstrapping for the mutual information)} 
\label{sec:regular-AL-proof}

A bound on $\mutinf{L}{L^c}$ can be derived from a bound on
$\maxinfeps{L}{L^c}$ as follows.  We use \Thm{thm:bootstrapping}
with $\eps = (\log d_L)^{-1}$, which is possible under our
assumption that we have an AGSP with $D^2\cdot \Delta \le c_0(\log
d_L)^{-16}$. Let $\rho_\eps\in B_\eps(\rho)$ be the (sub-)state that
minimizes $\maxinfeps{L}{L^c}_\gs$, i.e. $\rho_\eps \preceq t
\sigma_L\otimes \sigma_{L^c}$ where $t\EqDef
2^{\maxinfeps{L}{L^c}_\gs}$. Define the normalized state
$\hat\rho_\epsilon\EqDef \frac{\rho_\eps}{\Tr{\rho_\eps}}$ such that
$\hat \rho_\eps \preceq \frac{t}{\Tr{\rho_\eps}}\sigma_L\otimes
\sigma_{L^c}$. Notice that $\hat \rho_\eps$ is now in
$B_{2\eps}(\Omega)$, which can be shown using triangle inequality.
Note that $\Tr{\rho_\eps}\geq \Tr{\gs}-\norm{\gs-\rho_\eps}_1\geq
1-\eps$, which is in-fact larger than $1/2$ for $|L|> 1$, thus $\hat
\rho_\eps\preceq 2t \sigma_L\otimes \sigma_{L^c}$ and
$\maxinf{L}{L^c}_{\hat\rho_\eps}\leq 1+\maxinf{L}{L^c}_{\rho_\eps}$.
Using the second inequality of Fact~\ref{clm:D_inequality},
\begin{align*}
    \mutinf{L}{L^c}_{\hat \rho_\eps} \le \maxinf{L}{L^c}_{\hat
\rho_\eps} \le \maxinfeps{L}{L^c}_\gs+1 .
\end{align*}
We use the continuity of
mutual information (Fact~\ref{fact:cont-of-D}) to claim that
$|\mutinf{L}{L^c}_\gs - \mutinf{L}{L^c}_{\hat\rho_\eps}|\le 3 \cdot
\eps \cdot \log d_L + 3$, which implies
\begin{align*}
  \mutinf{L}{L^c}_\gs \le \mutinf{L}{L^c} _{\hat\rho_\eps} 
    + 3 \eps \cdot  \log d_L + 3 
  \le \maxinfeps{L}{L^c}_{\gs}  + 3\eps\cdot\log d_L + 4.
\end{align*}
Using \autoref{thm:bootstrapping}, the upper bound becomes 
\begin{align*}
  \mutinf{L}{L^c}_\gs \le 2\log D + 12\log(\log d_L/\eps) 
    + 3\eps\log d_L + O(1) .
\end{align*}
Recalling that $\eps=(\log d_L)^{-1}$, we get $\mutinf{L}{L^c}_\gs
\le 2\log D + 24\log\log d_L+ O(1)$.

\subsection{Proof of Corollary~\ref{corol:1D-AL} --- Area law for 
the maximally-mixed ground-state in 1D}
\label{sec:1D-AL-proof}

Let $\eps>0$, and consider a bi-partition $L\cup L^c$ of the line.  
Using Fact~\ref{fact:1D-AGSP}, we consider an AGSP $K(\ell,s)$ for
this bi-partition and use $\ell=s^2$. Then
\begin{align*}
  \Delta &= e^{-\bOmega{\gamma^{1/2}s^{3/2}}}, &  
  D &= e^{\bigO{s\log (sd)}},
\end{align*}
and so
\begin{align} \label{eq:almost-good-AGSP}
  2\log D + \log\Delta = \bigO{s\log (sd)} 
    - \bOmega{\gamma^{1/2}s^{3/2}} .
\end{align}
To impose the condition $D^2\cdot \Delta \le 
c_0\left(\frac{\eps}{\log d_L}\right)^8 =
c_0\left(\frac{\eps}{|L|\log d}\right)^8$ we need the
RHS of \eqref{eq:almost-good-AGSP} to be at most
$\log(c_0)-8\log\Big(\frac{|L|\log(d)}{\eps}\Big)$. 
For this to hold, it suffices to impose the following two conditions:
\begin{align*}
    \gamma^{1/2}s^{3/2} = \bigO{s\log(sd)}
\end{align*}
and
\begin{align*}
    \gamma^{1/2}s^{3/2} = \bigO{\log\Big(\frac{|L|\log
d}{\eps}\Big)}.
\end{align*}
The first condition is satisfied by
choosing $s=\bigO{\log^2(d/\gamma)/\gamma}$ (see \cRef{ref:Arad2013-1DAL}),
while the second condition is achieved by choosing
$s=\bigO{\frac{\log^{2/3}(|L|\log(d)/\eps)}{\gamma^{1/3}}}$. 
We can therefore satisfy the bootstrapping condition of
\Thm{thm:bootstrapping} by setting $s$ to the larger of
the two values.
If the second choice exceeds the first, that is, when
\begin{align} \label{eq:s-condition}
    \frac{\log^2(d/\gamma)}{\gamma} = \bigO{\frac{\log^{2/3}(|L|\log(d)/\eps)}{\gamma^{1/3}}}
    \tab & \Leftrightarrow &
    \frac{\log^3(d/\gamma)}{\gamma} = \bigO{\log (|L|\log (d)/\eps)} ,
\end{align}
which is expected in gapped Hamiltonians ($\gamma=\bigO 1$) with a constant qudit dimension ($d=\bigO 1$), 
$\log (D)$ becomes
\begin{align*}
    \log (D) & = \bigO{s\log(sd)} \\
    & =\bigO{\frac{\log^{2/3}(|L|\log(d)/\eps)}{\gamma^{1/3}} \cdot \log\left( 
    \frac{d}{\gamma^{1/3}}\cdot\log^{2/3}(|L|\log(d)/\eps)\right)} \\
    & =  \bigO{\frac{\log (d/\gamma)}{\gamma^{1/3}}\cdot\log^{2/3}(|L|\log (d)/\eps)} +
    \frac{1}{\gamma^{1/3}}\bigtO{\log^{2/3}(|L|\log(d)/\eps)},
\end{align*}
where in the last move we rewrote the logarithm of the product as a sum of logarithms and used 
$\log (d/\gamma^{1/3}) \le \log(d/\gamma)$.
We note that RHS of \Eq{eq:s-condition} implies $\log(d/\gamma)= \bigO{\gamma^{1/3}\log^{1/3}(|L|\log(d)/\eps)}$,
which clarifies the resulting expression for $\log (D)$:
\begin{align*}
    \log(D) & = \bigO{\log (|L|\log (d)/\eps)} +
    \frac{1}{\gamma^{1/3}}\bigtO{\log^{2/3}(|L|\log (d)/\eps)}
    =
    \bigO{\gamma^{-1/3}\log (|L|\log (d)/\eps)} .
\end{align*}

By \Thm{thm:bootstrapping}, the $\eps$-smoothed max-mutual
information in the maximally mixed ground state is bounded by
\begin{align*}
  \maxinfeps{L}{L^c}_{\gs} \le 2\log D + 12\log (|L|\log(d)/\eps) + O(1) 
   = \bigO{\gamma^{-1/3}\log (|L|\log (d)/\eps)}.
\end{align*}

Using the same argument, we choose
$\eps=(|L|\log (d))^{-1}$ so that
$D^2\cdot \Delta\le c_0\cdot (|L|\cdot\log(d))^{-16}$ and $\log
D=\bigO{\gamma^{-1/3}\log(|L|\log(d))}$, and by Corollary~\ref{corol:mutual-info},
\begin{align*}
  \mutinf{L}{L^c}_{\gs} = \bigO{\gamma^{-1/3}\cdot\log(|L|\log (d))} .
\end{align*}

Now we consider the case where the first choice for $s$ dominates, namely, 
\begin{align} \label{eq:s-condition-2}
    \frac{\log^2(d/\gamma)}{\gamma} = \bOmega{\frac{\log^{2/3}(|L|\log(d)/\eps)}{\gamma^{1/3}}}
    \tab & \Leftrightarrow &
    \frac{\log^3(d/\gamma)}{\gamma} =  \bOmega{\log (|L|\log (d)/\eps)}.
\end{align}
Here we get
\begin{align*}
    \log (D) & = \bigO{s\log(sd)} 
    = \bigO{\frac{ \log^3(d/\gamma)}{\gamma}},
\end{align*}
and similarly
\begin{align*}
    \maxinfeps{L}{L^c}_{\gs} & \le 2\log D + 12\log (|L|\log(d)/\eps) + O(1) 
   = \bigO{\frac{ \log^3(d/\gamma)}{\gamma}},\\
     \mutinf{L}{L^c}_{\gs} & = \bigO{\frac{ \log^3(d/\gamma)}{\gamma}} .
\end{align*}

\subsection{Proof of Corollary~\ref{corol:2D-AL} --- Area law for 
the maximally mixed ground state in 2D}
\label{sec:2D-AL-proof}

We let $\eps>0$ and consider a vertical bi-partition of the lattice
$L\cup L^c$ such that $\partial L$ is a vertical line.  Using
Fact~\ref{fact:2D-AGSP}, we consider an AGSP $K$ with respect to
this bi-partition such that $D^2\cdot \Delta\leq 1/2$ and $\log(D) =
{|\partial L|}^{1+\bigO{\log^{-1/5}|\partial L|}}$. Let $\ell$ be an
integer such that $2^{-\ell} =
\bTheta{\left(\eps/|L|\right)^8}$, i.e. $\ell = \bTheta{\log(|L|/\eps)}$, and
define a new AGSP by $K_\ell\EqDef K^\ell$ with corresponding
parameters $(D_\ell,\Delta_\ell)$.

We claim that 
\begin{align}
    D_\ell \leq D^\ell \;, &  &\Delta_\ell&\leq \Delta ^\ell.
\end{align}
The first property follows from sub-multiplicativity of the Schmidt
rank, e.g. for $K=\sum_{i=1}^D A_i\otimes B_i$, then $K^\ell =
\sum_{i_1=1}\dots \sum_{i_\ell=1}(A_{i_1}\dots A_{i_\ell})\otimes
(B_{i_1}\dots B_{i_\ell})$, so as \Def{def:op-SR} implies,
$\SR{L}{R}_{K^\ell}\leq D^\ell$. The second property is easily can
be seen by $K^\ell (\Id-\Pi_{gs})(K^\ell)^\dagger \leq \Delta\;
K^{\ell-1}(\Id-\Pi_{gs})(K^{\ell-1})^\dagger \leq \Delta^2 \dots
\leq \Delta^\ell(\Id-\Pi_{gs})$.  Therefore, by our choice of
$\ell$, we find that $D_\ell^2\cdot\Delta_\ell\leq (\Delta\cdot
D)^\ell\leq c_0 (\epsilon/|L|)^8$.

By \Thm{thm:bootstrapping}, the $\eps$-smoothed max-mutual
information in the maximally mixed ground state is bounded by
\begin{align*}
  \maxinfeps{L}{L^c}_{\gs} 
    &\le 2\log D_\ell + 12\log (|L|/\eps) + \bigO{1} \\
    &\le 2\ell \log D + 12\log (|L|/\eps) + \bigO{1} \\
    &= O\big(\log(|L|/\eps) \cdot\log D\big) .
\end{align*} 
 
Following the same argument as
in~Corollary~\ref{corol:1D-AL}, we choose $\eps=1/|L|$ so that
$D^2\cdot \Delta\le c_0\cdot |L|^{-16}$, by
Corollary~\ref{corol:mutual-info} and $\log(D) = {|\partial
L|}^{1+\bigO{\log^{-1/5}|\partial L|}}$. We get
\begin{align}
\label{eq:with-logL}
  \mutinf{L}{L^c}_{\gs} 
    = \bigO{ {|\partial L|}^{1 +
      \bigO{\log^{-1/5}|\partial L|}}\cdot \log|L|} .
\end{align}

For a square lattice where $\log(|L|)\leq 2 \log|\partial L|$, we
get that $\log |L|\leq {|\partial L|}^{\log^{-1/5}|\partial L|}$and
therefore we can absorb the $\log|L|$ factor in
\Eq{eq:with-logL} into $\bigO{ {|\partial L|}^{1 +
\bigO{\log^{-1/5}|\partial L|}}}$ and get
\begin{align*}
  \maxinfeps{L}{L^c}_{\gs} & 
   = \bigO{\log(1/\eps) \cdot {|\partial L|}^{1 + 
     \bigO{\log^{-1/5}|\partial L|}}} ,\\
   \mutinf{L}{L^c}_{\gs} 
     &= \bigO{ {|\partial L|}^{1 
       + \bigO{\log^{-1/5}|\partial L|}}} .
\end{align*}

\section{Low Schmidt rank and tensor network
approximations}
\label{Sec:approx}

This section is divided into two parts.
In the first part, \Sec{sec:lowSR}, we prove \Thm{thm:LowSR},
demonstrating a purification for the maximally mixed ground state with low Schmidt-rank approximations.
In the second part, \Sec{sec:TN}, we show how in one dimensional systems, \Thm{thm:LowSR} can be used to derive a tensor network approximation for the purification.
This, in turn, yields a similar structure for the maximally mixed ground state after tracing out the ancillary system.

\subsection{Proof of \Thm{thm:LowSR} --- Low Schmidt-rank approximation} \label{sec:lowSR}

The idea in the proof is to apply the AGSP
to the product state that saturates the area-law bound derived in
\Eq{eq:max-AL} ($\rho_\eps \preceq t \sigma_L\otimes \sigma_R$ where as before $R=L^c$),
which brings it closer to the ground state $\gs$. 
The low Schmidt rank of the resulting state will follow from choosing a well suited AGSP.
The analysis here is similar to \Thm{thm:bootstrapping}, and involves
the competition between the increasing Schmidt rank and the rate of convergence to the ground state.
Technically, we perform steps that are similar to the ones taken in the proof of \Thm{thm:bootstrapping},
demonstrating how the decrease in norm of the resulting state $K(\sigma_L\otimes \sigma_R)K^\dagger$
 overtakes the maximum information (the pre-factor $t$). 
Therefore, to relate the norm (which now involves GS overlap) and $t$, it is beneficial to work
in the extended space that involves the system + brothers, as done in the proof of
\Thm{thm:bootstrapping}. 
This achieves a low Schmidt-rank state which is close to the maximally mixed ground state.
The key difference from the proof of \Thm{thm:bootstrapping} lies in using a multiplicative symmetrization of the AGSP,
an operator that still satisfies the properties of an AGSP.
Doing this enables us to bound the Schmidt rank of the
\textbf{square root} rather than the state itself.
Finally, we use the fact that vectorizing the square root of 
a density operator yields a purification (see \Sec{sec:Vec}).

The following lemma contains the main technical steps
of the proof of \Thm{thm:LowSR} and,
in particular, establishes the key argument of the theorem
on the square root of the maximally mixed ground state. 

\begin{lemma}[Low Schmidt-rank approximation for the square root]
\label{lem:LowSR}  
  Let $\eps>0$, and let $H=\sum_i h_i$ be a local Hamiltonian on
  some lattice of qudits with a maximally-mixed ground state $\gs$.
   Under the same conditions in \Thm{thm:LowSR}, then there exists a Hilbert space
  $\mcH_B$ and an extension $\gs_{A}\mapsto \gs_{AB}$ such that
  $\gs_A=\Ptr{\gs_{AB}}{B}$, and for any bi-partition of the lattice
  $A=L\cup L^c$, there is a state $\gs_\eps\in \mcD(AB)$ for which:
  1.~$\norm{\gs_{AB}-\gs_\eps}_1\le\eps$. 2. The Schmidt rank of
  $\sqrt{\gs_\eps}$ with respect to the $L:
  L^cB$  bi-partition satisfies 
  \begin{align*} \label{eq:bound-Omega_eps}
      \mySR (\sqrt{\gs_\eps})\le 
    	  49 D^2 \cdot \Big(\frac{\log d_L}{\eps}\Big)^2.
  \end{align*}
\end{lemma}

Our motivation for considering the square root arises from several
key reasons. First, it provides a stronger condition than having low
Schmidt rank for the state itself, which follows from the bound
$\mySR(O)\le\mySR(\sqrt{O})^2$ (see \Def{def:op-SR}). Additionally,
having $\gs$ and $\gs_\eps$ close in $L_1$ norm also implies that
their square roots are close in $L_2$ norm. As the square root of a
state is closely related to its purification (see
Fact~\ref{fact:vec}), \Lem{lem:LowSR} implies the results of \ref{thm:LowSR}, namely, there exists a
purification of the maximally-mixed ground state that can be
approximated by a pure state of low Schmidt rank.  In \Sec{sec:TN}, we will combine this result with the
Young-Eckart theorem, enabling us to truncate the Schmidt rank
with respect to a given cut in the lattice while maintaining
controlled proximity.

\begin{proof}[ of \Thm{thm:LowSR} using \Lem{lem:LowSR}] 
    We apply \Lem{lem:LowSR} with parameter
    $\eps$, to get an extending state $\gs_{AB}$
    such that for any bi-partition $A=L:L^c$, there is a
    state $\gs_{\eps}$ on $AB$ where $\norm{\gs_{\eps}-\gs_{AB}}_1\le \eps$ and whose Schmidt rank satisfies \Ineq{eq:bound-Omega_eps}.
    Recall that for a density matrix $\rho_{AB}$, the vectorized square root 
    $\dket{\sqrt \rho }_{A\tilde A B \tilde B}$ is a purification (Fact~\ref{fact:vec}).
    Moreover, the purification $\dket{\sqrt{\gs_\eps}}_{A\tilde A B \tilde B}$ has bounded Schmidt rank:
    \begin{align*}
        \SR{L\tilde L}{R\tilde R B \tilde B}_{\dket{\sqrt{\gs_\eps}}}
        =\SR{L}{RB}_{\sqrt{\gs_\eps}} .
    \end{align*}
    Now we use Facts~\ref{fact:vec} (bullet 2) and \ref{fact:sqrt} and to claim that
    \begin{align} 
        \norm{\ket{\gs}-\dket{\sqrt{\gs_\eps}}}^2
        = \norm{\sqrt{\gs_{AB}}-\sqrt{\gs_\eps}}_2^2
        \le \eps .
        \end{align}
        Choosing $E= B \tilde B$, 
        $\ket \gs_{A\tilde AE} \EqDef\dket{\sqrt \rho }_{A\tilde A B \tilde B}$, and
        $\ket {\psi^{(L)}}_{A\tilde A E} \EqDef\dket{\sqrt \rho }_{A\tilde A B \tilde B}$ concludes the proof.
\end{proof}

\begin{proof}[ of \Lem{lem:LowSR}]
Let $A=L\cup R$ be a bi-partition of the lattice. Let $\eps>0$ and set $\delta=\eps/44$.
Apply \Thm{thm:bootstrapping} with parameter $\delta$ and the bi-partition $L:R$.
Let $\rho$ and $\sigma=\sigma_L\otimes \sigma_R$ denote
the sub-state and product state, that achieves the smooth max information, respectively, as provided in the theorem, i.e.
\begin{align*}
    \rho  \preceq t\sigma, \tab
     \norm{\rho-\gs}_1  \le \delta \tab
    \log(t)  = \maxinfdel{L}{R}_\gs,
\end{align*}
where it is guaranteed by \Thm{thm:bootstrapping}
that $t=2^{c_1} D^2 \cdot \Big(\frac{\log d_L}{\delta}\big)^{12}$.
Note that here $\gs$ refers to the original maximally-mixed ground state and not the extension of it.
We now perform similar steps as in the proof of \Thm{thm:bootstrapping}.
First, we apply \Lem{lem:discretiztion} on $\rho \preceq t\sigma$ with
parameter $\eps$ to achieve $\hat \rho \preceq 2t\tilde
\sigma_L\otimes \sigma_R$, where
$|\Spec(\tilde\sigma_L)|\le 7 \log(d_L/\eps)$ and
$\norm{\hat \rho-\rho}_1\le 2(\eps/d_L)^2$. Now, we extend the
resulting states to states on a larger Hilbert space using
\Lem{lem:flat} with parameter $\delta$ to achieve 
\begin{align*}
    \rho'_{AB}  \preceq t'\sigma'_{AB},\tab
     \norm{\rho'_{A}-\hat \rho}_1 \le \delta ,\tab
     \SR{L}{RB}_{\sigma'}  \le 7\log(d_L/\eps),
\end{align*}
where $\sigma'_{AB}=\frac{\Pi_{\sigma'}}{d_\sigma}$ is a flat state, and 
$t'=2^{\dmax{\rho'_{AB}}{\sigma'_{AB}}} \le t\cdot \frac{64}{\delta^2}$.

Let 
\begin{align} \label{eq:extended_GS}
    \gs_{AB}\EqDef \frac{1}{d_B}\Pi_{gs}\otimes\Pi_r
\end{align} 
be the extension of the ground state to $AB$ as defined in
\Eq{eq:extended-AGSP}.
Let $K$ be the $(D, \Delta)$-AGSP which was assumed \textit{a priori} in the theorem statement to satisfy the condition
\begin{align*}
    D^2\Delta \leq c_0 \cdot\Big(\frac{\delta}{\log d_L}\Big)^8,
\end{align*}
and consider the extended AGSP $K_{AB}=K_A\otimes \Pi_r$
as defined in \Eq{eq:extended-AGSP}, serving as  an AGSP on the image of $\gs_{AB}$. 
Now we define the following symmetrized version of it
\begin{align*}
    \tilde K _{AB} \EqDef \Pi_{\sigma'} K_{AB}^\dagger K_{AB}
\end{align*}
and apply to both sides of
$\rho'_{AB} \preceq t'\sigma'_{AB}$ to get
\begin{align} \label{Ineq:almost-final-ineq}
    \tilde K_{AB}\rho'_{AB} \tilde K^\dagger_{AB}  \preceq t' \tilde K_{AB} \sigma'_{AB} \tilde K_{AB}^\dagger =
    t' \Tr{\tilde K_{AB} \sigma'_{AB} \tilde K_{AB}^\dagger} \gs_\eps,
\end{align}
where we set $\gs_\eps\EqDef \frac{\tilde K_{AB}
\sigma'_{AB} \tilde K_{AB}^\dagger}{\Tr{\tilde K_{AB} \sigma'_{AB}
\tilde K_{AB}^\dagger} }$. 
We analyze the trace similarly to Bullet~\ref{clm:quantities-3} of Claim~\ref{clm:quantities}:
\begin{align*}
    \Tr{\tilde K_{AB} \sigma'_{AB} \tilde K_{AB}^\dagger}
    \le \Tr{ (K_{AB}^\dagger K_{AB}) \sigma'_{AB} (K_{AB}^\dagger K_{AB} )}
    \le \Tr{\Pi'_{gs} \sigma'_{AB}} + \Delta^2.
\end{align*}
where in the first step we got rid of $\Pi_\sigma$ using the fact
that $\Tr{\Pi \rho}\le \Tr{\rho}$ for any PSD operator $\rho$ and projector $\Pi$, and in the
second we separated the trace to the extended ground state part and the complement as done in Claim±\ref{clm:quantities}.
We adopt the fact that $\sigma_{AB}'$ is flat
 to write
\begin{align*}
    \Tr{\Pi'_{gs} \sigma'_{AB}}
    =\frac{1}{d_\sigma}\Tr{\Pi'_{gs} \Pi_{\sigma'}}
    \le \frac{1}{d_\sigma}\Tr{\Pi'_{gs} }
    = \frac{d_B}{d_\sigma} 
    = \frac{1}{d_\sigma \lambda_{\max}(\rho'_{AB})}
    = 2^{-\dmax{\rho'_{AB}}{\sigma'_{AB}}},
\end{align*}
where the second last move is due to $\lambda_{\max}(\rho'_{AB})=1/d_B$ following \Lem{lem:flat}, 
and the last move is due to \Lem{lem:Dmax_flat_2}. 
Note that the last term is just $1/t'$, so that
\Ineq{Ineq:almost-final-ineq} becomes
\begin{align} \label{def:eta}
   \eta_{AB}\EqDef \tilde K_{AB}\rho'_{AB} \tilde K^\dagger_{AB}  \preceq 
    (1+t' \Delta^2) \gs_\eps
    = (1+\tilde\delta) \gs_\eps
\end{align}
where we defined $\tilde \delta \EqDef t'\Delta^2$.
Combined with \Lem{lem:short-dist}, we get that
\begin{align} \label{eq:for-the-trace}
    \norm{\eta-\gs_\eps}_1 \le 2\tilde\delta+ (1-\Tr{\eta}).
\end{align}
Later, we will verify that $\tilde \delta $ is sufficiently small,
ensuring that \Eq{def:eta} implies closeness of $\gs_\eps$ and $\eta$.

\par To finish the proof, it remains to show two statements: \newline
\indent 1. Show that indeed $\norm{\gs_{AB}-\gs_\eps}_1\le\eps$. \newline
\indent 2. Show that $\sqrt{\gs_\eps}$ has low Schmidt rank. \\

We begin with the first statement; we do this by first showing that $\eta$ is close to $\gs_{AB}$,
and then, together with \eqref{eq:for-the-trace}, use triangle inequality to conclude that $\gs_\eps$ is close to $\gs_{AB}$.
First, we use triangle inequality with $\Pi_{\sigma'}\gs_{AB}\Pi_{\sigma'}$: 
\begin{align*}
    \norm{\eta-\gs_{AB}}_1 & \le
    \norm{\eta - \Pi_{\sigma'}\gs_{AB}\Pi_{\sigma'}}_1 
     + \norm{\gs_{AB}-\Pi_{\sigma'}\gs_{AB}\Pi_{\sigma'}}_1 .
    %
\end{align*}
To handle the first term in the RHS, we insert the definition of $\eta$ from \eqref{def:eta},
and use the fact that $K$ and $ K^\dagger$ fix the ground state $\gs_{AB}$, and
$\norm{K},\norm{\Pi_{\sigma'}}\le 1$ to achieve
\begin{align*}
    \norm{\eta - \Pi_{\sigma'}\gs_{AB}\Pi_{\sigma'}}_1 
    = \norm{\Pi_{\sigma'} K_{AB}^\dagger K_{AB}(\rho'_{AB}-\gs_{AB})K_{AB}^\dagger K_{AB}\Pi_{\sigma'}}_1 
    \le  \norm{\rho'_{AB}-\gs_{AB}}_1 .
\end{align*}
For the second term, we use triangle inequality with $\rho'$ and the fact that $\im(\rho'_{AB})\subseteq \im(\sigma_{AB})$, i.e. $\rho'_{AB} = \Pi_{\sigma'}\rho'_{AB}\Pi_{\sigma'}$, to conclude 
\begin{align*}
    \norm{\gs_{AB}-\Pi_{\sigma'}\gs_{AB}\Pi_{\sigma'}}_1 
    & = \norm{\gs_{AB}-\rho'_{AB}}_1 
    +\norm{\Pi_{\sigma'}\gs_{AB}\Pi_{\sigma'}-\rho'_{AB}}_1 \\
    & =\norm{\gs_{AB}-\rho'_{AB}}_1 
    +\norm{\Pi_{\sigma'}(\gs_{AB}-\rho'_{AB})\Pi_{\sigma'}}_1 \\
    & \le 2\norm{\gs_{AB}-\rho'_{AB}}_1 .
\end{align*}
So we got that 
\begin{align*}
    \norm{\eta-\gs_{AB}}_1 & \le
    3 \norm{\gs_{AB}-\rho'_{AB}}_1.
\end{align*}
Further calculations, that will be presented below, produce the following:
\begin{claim} \label{clm:gs-rho}
    $\norm{\gs_{AB}-\rho'_{AB}}_1 \le 7\delta $.
\end{claim}
Using this claim, we get
$\norm{\eta-\gs_{AB}}_1  \le 3 \cdot 7 \delta=21\delta$,
and thus, using \Ineq{eq:for-the-trace}:
\begin{align*}
    \norm{\gs_\eps-\gs_{AB}}_1
    & \le \norm{\gs_\eps-\eta}_1
    + \norm{\eta- \gs_{AB}}_1 
    \\
    & \le 2\tilde\delta +(1-\Tr{\eta})
    +\norm{\eta - \gs_{AB}}_1 
    \\
    & \le 2\tilde\delta + 2\norm{\eta - \gs_{AB}}_1 \\
    & \le 2\tilde \delta+2\cdot 21\delta ,
\end{align*}
where the in first inequality we used triangle inequality,
in the second we used \eqref{eq:for-the-trace},
and in the third we used inverse triangle inequality $\Tr{\eta}\ge \Tr{\gs_{AB}}-\norm{\eta-\gs_{AB}}_1$.

We conclude by showing that $\tilde  \delta\le\delta$,
resulting in $\norm{\gs_\eps-\gs_{AB}}_1\le 44\delta=\eps$.
This follows from the specific choice of AGSP in
the theorem, for which $D^2\cdot\Delta\le c_0 \Big(\frac{\delta}{\log d_L}\Big)^8$ with $c_0=10^{-16}$,
and from the parameters choice in the proof, $t'\le t\frac{64}{\delta^2}$
and $t= D^2\cdot \Big(\frac{\log d_L}{\delta}\Big)^{12} \cdot 2^{c_1}$ for $c_1\approx 76$.
\begin{align*}
    \tilde \delta & = \Delta^2 t' \le \Delta^2 t\frac{64}{\delta^2}\\
    & = \Delta^2\cdot D^2\cdot \Big(\frac{\log d_L}{\delta} \Big)^{12}2^{c_1}\cdot\frac{64}{\delta^2}  \\ 
    & \le 2^{c_1+6}(\Delta\cdot D^2)^2\cdot \Big(\frac{\log d_L}{\delta} \Big)^{12}\cdot\frac{1}{\delta^2} \\
    & \le 2^{c_1+6}(c_0)^2 \Big(\frac{\delta}{\log d_L} \Big)^{16}\cdot
    \Big(\frac{\log d_L}{\delta} \Big)^{12}\cdot\frac{1}{\delta^2} \\
    & \le 2^{c_1+6}\cdot(c_0)^2 \Big(\frac{\delta}{\log d_L} \Big)^{2}
    \le \delta ,
\end{align*}
where in the first inequality we used $D\ge 1$, then we used the condition on the AGSP,
and finally, $\log d_L\ge 1$ and the fact that $2^{c_1+6}(c_0)^2\ll 1$ and $\delta<1$.

After showing that $\gs_\eps$ is $\eps$-close to $\gs_{AB}$ in trace norm,
we are left to address the Schmidt rank of $\sqrt{\gs_\eps}$.
Recall that 
\begin{align*}
    \Omega_\eps \propto 
    (\Pi_{\sigma'} K_{AB}^\dagger K_{AB})\sigma'( K_{AB}^\dagger K_{AB}\Pi_{\sigma'}) .
\end{align*} 
Considering $\sigma'$ being flat (due to \Lem{lem:flat}), i.e. 
$\sigma'=\Pi_{\sigma'}/d_\sigma$ where $\Pi_{\sigma'}$ is a projector, we get
\begin{align*}
    \Omega_\eps \propto 
    (\Pi_{\sigma'} K_{AB}^\dagger K_{AB}\Pi_{\sigma'})
    (\Pi_{\sigma'} K_{AB}^\dagger K_{AB}\Pi_{\sigma'}) ,
\end{align*}
That is, $\sqrt{\gs_\eps}\propto \Pi_{\sigma'} K_{AB}^\dagger K_{AB}\Pi_{\sigma'}$.
This expression allows us to upper-bound the Schmidt rank with respect to the bi-partition $L:RB$ in the following manner
\begin{align} \label{eq:almost-SR-bound}
    \mySR(\sqrt{\gs_\eps})\le
    \mySR(\Pi_{\sigma'})^2\cdot \mySR(K_{AB})
    \mySR(K_{AB}^\dagger) .
\end{align}
which is given due to the sub-multiplicativity of the operator Schmidt rank.
Using $K_{AB}=K_A \otimes\Pi_r$, we get
\begin{align*}
    \SR{L}{R B}_{K_{AB}}= 
    \SR{L}{R}_{K_{A}}
    = D.
\end{align*}
Recalling that $\sigma'\propto \Pi_{\sigma'}$, $\mySR(\Pi_{\sigma'})=\mySR(\sigma')
\le 7\log( d_L/\eps)$, and the desired upper-bound on the Schmidt rank is obtained from \Eq{eq:almost-SR-bound}.

To complete the proof of \Thm{lem:LowSR}, it remains to show that
the extension is independent of the choice of bi-partition $L\cup R$.
In the proof, we fixed a bi-partition and then applied \Lem{lem:flat} tailored 
specifically to it. As a result, the dimension of
$\mcH_B$, and correspondingly the 
extension of the ground state (given in~\Eq{eq:extended_GS}) may vary for different bi-partitions. 
We overcome this problem by referring to Bullet 3 of~\Lem{lem:flat}, which tells us
that given an extension with $d_B=\dim(\mcH_B)$, 
one can also consider an extension with $\tilde d_B$ which is a multiple of $d_B$.
Thus, we unify all extensions by replacing each $d_B=\dim(\mcH_B)$ 
associated with a given bi-partition to the least common multiple of all $\{d_B\}_{A=L\cup R}$,
i.e., the smallest common multiple of all $d_B$ arising from different bi-partitions.
Doing so will not change the proof, as guaranteed by \Lem{lem:flat}, nor the results, that are independent of $d_B$.
Moreover, one can see that the extension in~\Eq{eq:extended_GS} depends solely on the dimension of $\mcH_B$.
Thus, the extension is independent of the chosen bi-partition.

\begin{proof}[ of Claim~\ref{clm:gs-rho}]
To show that indeed $\rho'_{AB}$ is close to $\gs_{AB}$, we need to
consider an intermediate state. Recall the state $\hat\rho$ obtained
from \Lem{lem:discretiztion}. Let $\hat \rho=\sum_i a_i
\ketbra{a_i}{a_i}$ be a spectral decomposition, where $a_i$ are
decreasingly ordered.  The intermediate state is defined by it's
flat extension to the brothers space (similarly as in the beginning
of the proof of \Lem{lem:flat}): 
\begin{align*}
    \hat \rho _{AB} = \frac{1}{d_B}\sum_i \ketbra{a_i}{a_i} \otimes\Pi^B_{d_B\cdot a_i}
\end{align*}
where $\Pi^ B_{d_B\cdot a_i}=\sum\limits_{m=1}^{d_B\cdot a_i} \ketbra m m _B$.
Triangle inequality gives
\begin{align} \label{eq:what-we-need}
     \norm{\gs_{AB}-\rho'_{AB}}_1
    & \le \norm{\gs_{AB}-\hat\rho_{AB}}_1 
     +  \norm{\hat \rho_{AB}-\rho'_{AB}}_1 .
     \end{align}
The second term is evident from \Lem{lem:flat}, which tells that not
only $\norm{\hat\rho_A-\rho'_A}_1\le \delta$,
but also $\norm{\hat \rho _{AB} - \rho'_{AB}}_1\le \delta$. 
Now we handle the first term $\norm{\gs_{AB}-\hat \rho_{AB}}_1$. To show this, we define an additional intermediate state
\begin{align*}
    \rho_{\mathrm{int}} \EqDef \frac{1}{d_B}\sum_{i=1}^r \ketbra{a_i}{a_i} \otimes \Pi_r
\end{align*}
and use triangle inequality to achieve
\begin{align*}
    \norm{\gs_{AB}-\hat \rho_{AB}}_1 &
    \le \norm{\gs_{AB}-\rho_{\mathrm{int}}} _1
    +\norm{\hat \rho_{AB}-\rho_{\mathrm{int}}} _1\\
    & = \frac{1}{d_B}\norm{\big(\Pi_{gs}-\sum_{i=1}^r\ketbra{a_i}{a_i}\big)\otimes \Pi_{r}} _1
    +\frac{1}{d_B}\norm{\sum_i \ketbra{a_i}{a_i}\otimes (\Pi_{d_B \cdot a_i} - \Pi_{r})} _1\\
    & = \frac{1}{r}\norm{\Pi_{gs}-\sum_{i=1}^r\ketbra{a_i}{a_i}} _1
    +\frac{1}{d_B}\sum_i\norm{ \Pi_{d_B \cdot a_i} - \Pi_{r}} _1\\
    & \le \norm{\gs_A-\hat \rho_A} _1
    + \sum_i |a_i-1/r|
    +\sum_i|a_i -1/r|
\end{align*}
where in the second and third step we used the multiplicativity of $\norm{\cdot}_1$ under tensor product, and then in the final inequality, at the left part we used triangle inequality with $\hat \rho_A$, and at the right part we used the fact the brothers projectors are diagonal, so that $\norm{\Pi_\ell-\Pi_m}_1 = |m-\ell| $.
Notice that our specific choice of parameters yields
\begin{align*}
        \norm{\gs_A-\hat \rho_A}_1 \le 
         \norm{\gs_A-\rho}_1+\norm{\rho-\hat\rho_A }_1
         \le \delta+\delta=2\delta .
\end{align*}
Using Fact~\ref{fact:bhatia}, we obtain\footnote{Notice that we are implicitly considering $1/r$ on the first $r$ elements in the summation. In the remaining part, i.e. $i>r$, we set $1/r\mapsto 0$. This is also true in the derivation before where we write $\sum_i|a_i-1/r|$. } $\sum_i |a_i-1/r|\le 
 \norm{\hat \rho_A-\gs_A}_1\le 2\delta$.
So we got
\begin{align*}
    \norm{\hat \rho_{AB}-\gs_{AB}}_1
    \le \norm{\hat \rho_A-\gs_A} _1
    + 2\sum_i |a_i-1/r|
    \le 2\delta + 4\delta = 6\delta .
\end{align*}
Plugging to \eqref{eq:what-we-need} gives the desired bound.

\end{proof}

\subsection{Proof of Corollary~\ref{corol:TN} --- MPO approximation}
\label{sec:TN}

We now proceed to prove Corollary~\ref{corol:TN} and derive a matrix-product-operator (MPO) approximation for $\gs$.
To do so, we construct a matrix-product-state (MPS) approximation to the purification of the ground state, then trace out the ancilla (see \Fig{fig:MPO}).
The existence of such an MPS is guaranteed by the following lemma taken
from \cRef{ref:Verstrate2006MPS}, which analyzes the truncation error 
due to a repeated projection to the largest Schmidt states at each cut.
\begin{fact} [Lemma 1 from \cRef{ref:Verstrate2006MPS}] \label{lem:MPS}
    Let $\ket \psi$ be pure quantum state on $n$ sites of local dimension $d$.
    For each bi-partition $\{1\rightarrow k \}:\{k+1\rightarrow n \}$, let 
    $\epsilon_1^{(k)},\epsilon_2^{(k)},\ldots$ denote the eigenvalues of the reduced density matrix $\rho_{1\rightarrow k}$.
    There is an MPS $\ket{\psi_{MPS}}$ of bond dimension $D_k$ at the $k$-cut, such that 
    \begin{align*}
        \norm{\ket \psi - \ket{\psi_{MPS}}}^2
        \le 2\sum_{k=1}^{n-1}{\eps^{(k)}_{>D_k}},
    \end{align*}
    where $\epsilon^{(k)}_{>D_k}=\sum_{i>D_k}\limits\epsilon^{(k)}_i$.
\end{fact}

Control over the truncation error of the Schmidt coefficients
of the purified ground state is straightforward 
by combining \Thm{thm:LowSR} and the Young Eckart theorem: 

\begin{corol}[Truncation error] \label{corol:cutting-SR}
    Let $\eps>0$, and let $\ket\gs_{A\tilde A E}$ be the purification of the fully mixed ground state provided in \Thm{thm:LowSR}.
    Given a bi-partition of the physical lattice $A = L: R$, let $\lambda_1\geq \lambda_2 \geq \dots$ 
    denote the Schmidt coefficients of $\ket \gs$ with respect to the bi-partition $L\tilde L: R\tilde R E$. 
    Then $\{\lambda_i\}$ satisfy
    \begin{align*}
        \sum_{i>D_L} \lambda_i^2 \le \eps 
    \end{align*}
    for $D_L\EqDef \mySR(\psi^{(L)})$ satisfying \Ineq{eq:bound-Omega_eps}.
\end{corol}
\end{proof}


We are now ready to derive the MPO approximation for the 
purification of the fully mixed ground state of a 1D gapped local Hamiltonian.
\begin{proof}[ of Corollary~\ref{corol:TN} (Derivation of MPO)]
\begin{figure}[t]
    \centering
    \includegraphics[width=1\linewidth]{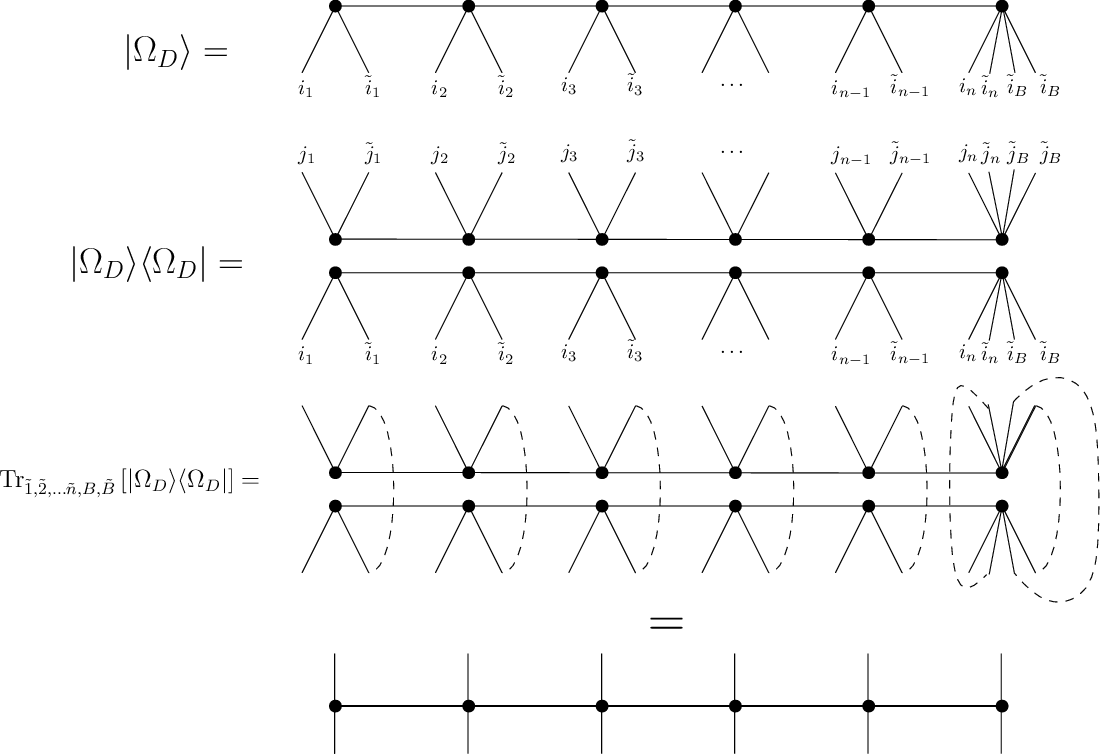}
    \caption{Tensor network structure of $\ket{\gs_D}_{A\tilde A B \tilde B}$, its density matrix $\ketbra{\gs_D}{\gs_D}_{A\tilde A B \tilde B}$ and its reduced matrix $\Psi= \Ptr{\ketbra{\gs_D}{\gs_D}}{\tilde A B \tilde B}$.}
     \label{fig:MPO}
\end{figure}    

Given $\eps>0$, apply Corollary~\ref{corol:cutting-SR} with parameter $\eps'=\eps^2/(8n)$ to get 
 a purification $\ket \gs _{A\tilde A E}$. 
 As claimed in Corollary \ref{corol:cutting-SR}, for each bi-partition of the 1D lattice $A=L:R=\{1\rightarrow k \}:\{k+1\rightarrow n \}$ where $k=1,\ldots, n-1$, the Schmidt coefficients $\{\lambda_i^{(k)}\}_i$ of $\ket \gs$ with respect to $L\tilde L : R\tilde R E$ satisfy 
 \begin{align*}
     \sum_{i>D_k} (\lambda_i^{(k)})^2 \le \eps^2/(2n)
 \end{align*}
 where 
 \begin{align*}
     D_k & \underbrace \EqDef \mySR(\ket{\psi^{(L)}}) 
     \leq_{\text{Corollary}~\ref{corol:cutting-SR}} 49 D^2 \left( \frac{|L|}{\eps'} \right)^2 \\ &
     \underbrace = _{\text{Corollary}~\ref{corol:1D-AL}}
     49  (k/\eps')^{\bigO{\gamma^{-1/3}}}  \frac{k^2 }{\eps'^2}
     = \poly(k/\eps') = \poly(n/\eps) ,
 \end{align*}
 where at the last step we inserted $\eps'=\eps^2/(8n)$ and $k\le n$.
Considering the fact that the squared Schmidt coefficients correspond to the eigenvalues of the reduced density matrix,
we apply \Lem{lem:MPS} to achieve an MPS $\ket{\gs_D}\in \mathcal H_{A\tilde A B\tilde B}$ with maximal bond dimension $D = \max_k D_k$ such that $\norm{\ket \gs-\ket{\gs_D}}^2 \le \eps^2/4$.
Notice that when we consider the MPS representation of $\ket {\gs_D}$, we look on the $n$'th qudit and system $E$ as a single entity,
namely, we associate a single tensor for both systems, as seen in Fig.~\ref{fig:MPO}.
Computing the reduced density matrix of $\ket{\gs_D}$ to $A$ achieves an MPO with bond dimension $D^2=\poly(n/\eps)$.
To demonstrate this, we write the MPS and MPO explicitly, as shown diagrammatically in Fig.~\ref{fig:MPO}.
First, we write the MPS from \Lem{lem:MPS}:
\begin{align*}
    \ket{\gs_D}
    = \sum_{\{i_k\},\{\tilde i_k\},i_B,\tilde i_B}
    \Tr{A_{i_1}^{\tilde i _1}\cdot A_{i_2}^{\tilde i _2} 
    \dots A_{i_{n-1}}^{\tilde i _{n-1}}\cdot A_{i_n,i_B}^{\tilde i _n,\tilde i_B}} 
    \ket{i_1,\ldots, i_n}_A \ket{\tilde i_1,\ldots ,\tilde i_n}_{\tilde A}
    \ket{i_B,\tilde i_B}_{B\tilde B} ,
\end{align*}
where each $A_{i_k}^{\tilde i _k}$ is a $D_{k-1}\times D_{k}$ matrix.
Then, taking the partial trace over $\tilde AB\tilde B$, we achieve the following expression:
\begin{align*}
    \Ptr{\ketbra{\gs_D}{\gs_D}}{\tilde A B\tilde B}
     & = \sum_{\{i_k\},\{j_k\},\{\tilde i_k\},i_B,\tilde i_B}
    \Tr{A_{i_1}^{\tilde i _1}\cdot A_{i_2}^{\tilde i _2} 
    \dots A_{i_{n-1}}^{\tilde i _{n-1}}\cdot A_{i_n,i_B}^{\tilde i _n,\tilde i_B}} \\
    & \; \;\; \;\;\;\;\;\;\;\;\;\;\;\;\;\;\;\;\;\;\;\;\;\;\; \cdot
    \overline{\Tr{A_{j_1}^{\tilde i _1}\cdot A_{j_2}^{\tilde i _2} 
    \dots A_{j_{n-1}}^{\tilde i _{n-1}}\cdot A_{j_n,i_B}^{\tilde i _n,\tilde i_B}} } 
     \;\ketbra{\{i_k\}}{\{j_k\}}_A \\
      & = \sum_{\{i_k\},\{j_k\}}
    \mathrm{Tr}\Big[\left(\sum_{\tilde i_1} A_{i_1}^{\tilde i _1}\otimes 
    \overline{A_{j_1}^{\tilde i _1}}\right)\cdot
    \left(\sum_{\tilde i_2} A_{i_2}^{\tilde i _2}\otimes 
    \overline{A_{j_2}^{\tilde i _2}}\right)
    \dots \left(\sum_{\tilde i_{n-1}} A_{i_{n-1}}^{\tilde i _{n-1}}\otimes 
    \overline{A_{j_{n-1}}^{\tilde i _{n-1}}}\right)\\
    & \cdot \left(\sum_{\tilde i_n,i_B,\tilde i_B} A_{i_n,i_B}^{\tilde i _n,\tilde i_B}\otimes 
    \overline{A_{j_n,i_B}^{\tilde i _n,\tilde i_B}}\right)\Big] 
     \ketbra{\{i_k\}}{\{j_k\}}_A \\
     & = \sum_{\{i_k\},\{j_k\}}
    \Tr{B_{i_1}^{j _1}\cdot
    B_{i_2}^{j _2}
    \dots B_{i_{n-1}}^{j_{n-1}} \cdot B_{i_n}^{j_n}}
     \ketbra{i_1,\ldots ,i_n}{j_1,\ldots ,j_n}_A .
\end{align*}
Here, each of the $B_{i_k}^{j_k} \EqDef \sum_{\tilde i_k} A_{i_k}^{\tilde i _k}\otimes 
    \overline{A_{j_k}^{\tilde i _k}}$ is a $(D_{k-1})^2\times (D_k)^2$ matrix for any $i_k,j_k=0,\ldots,d-1$.

    We finish by noting that $\Psi\EqDef \Ptr{\ketbra{\gs_R}{\gs_R}}{\tilde A B\tilde B}$ is indeed close to the ground state $\gs$, due to monotonicity of $\norm{\cdot}_1$ under partial tracing:
    \begin{align*}
        \norm{\Psi-\gs}_1 & \le 
        \norm{\ketbra{\gs_D}{\gs_D}_{AB\tilde A \tilde B}
        -\ketbra \gs \gs _{AB\tilde A \tilde B}}_1 \\
        & = \norm{(\ket {\gs_D}-\ket \gs)\bra{\gs_D} + \ketbra{\gs}{\gs_D} - \ket \gs (\bra{\gs}-\bra{\gs_D}) -\ketbra{\gs}{\gs_D}}_1 \\
        & \le \big(\norm{\ket{\gs}}+\norm{\ket{\gs_D}}\big)\cdot\norm{\ket \gs - \ket{\gs_D}} \le \eps,
    \end{align*}
    where in the first inequality we used monotonicity, in the second inequality
    we used triangle inequality and the fact that $\norm{\ketbra \phi \psi}_1= \norm \phi \norm \psi$, and in the last inequality we used $\norm{\ket \gs-\ket{\gs_D}}^2 \le \eps^2/4$.

\end{proof}

\vspace{2mm}

\section*{Acknowledgment} 

We thank A.~Anshu for illuminating discussions on this work.

The work of I.A. and R.J. is supported by the National Research Foundation, Singapore, through the National Quantum Office, hosted in A*STAR, under its Centre for Quantum Technologies Funding Initiative (S24Q2d0009) and the NRF grant
NRF2021-QEP2-02-P05 and the Ministry of Education, Singapore, under
the Research Centres of Excellence program. 
 This work was done in
part while R.J. was visiting the Technion-Israel Institute of
Technology, Haifa, Israel, and the Simons Institute for the Theory
of Computing, Berkeley, CA, USA. I.A.\ acknowledges the support of
the Israel Science Foundation (ISF) under the Individual Research
Grant No.~1778/17 and joint Israel-Singapore NRF-ISF Research Grant
No.~3528/20.

\bibliographystyle{ieeetr}
\bibliography{biblist.bib}

\end{document}